\documentclass[reprint,aps,showpacs,floatfix,superscriptaddress,nofootinbib]{revtex4-1}

\usepackage{graphicx}
\usepackage{amsthm}   
 \usepackage[bottom]{footmisc}
\usepackage{amsmath} 

\usepackage{amssymb}  
\usepackage{mathrsfs} 
\usepackage{stmaryrd} 
\usepackage{txfonts} 

\usepackage{enumerate}
\usepackage{appendix}

\newcommand{\mathsym}[1]{{}}
\newcommand{\unicode}[1]{{}}

\newcommand{\E}{{\mathbb{E}}}

\newcommand{\N}{{\mathbb{N}}}

\newcommand{\X}{{\mathcal{X}}}
\newcommand{\Y}{{\mathcal{Y}}}
\newcommand{\G}{{\mathscr{G}}}
\newcommand{\WW}{\Omega}
\newcommand{\PP}{{\mathcal{P}}}

\usepackage{color}  
\RequirePackage[dvipsnames,usenames]{xcolor}

\newcommand{\be}{\begin{equation}}
\newcommand{\bel}[1]{\begin{equation}\label{#1}}
\newcommand{\qe}{\end{equation}}
\newcommand{\ee}{\end{equation}}
\newcommand{\eeq}{\end{equation}}
\newcommand{\ba}{\begin{eqnarray}}
\newcommand{\ea}{\end{eqnarray}}

\newtheorem{theorem}{Theorem}

\newtheorem{corollary}[theorem]{Corollary}

\newtheorem{definition}{Definition}
\newtheorem{example}{Example}

\newtheorem{lemma}[theorem]{Lemma}

\newtheorem{proposition}[theorem]{Proposition}

\date{\today}                      

\begin{document}

\title{The free energy requirements of biological organisms; implications for evolution}

 \author{David H. Wolpert}
\altaffiliation{Massachusetts Insitutute of Technology}
 \altaffiliation{Arizona State University}
  \affiliation{Santa Fe Institute, 1399 Hyde Park Road, Santa Fe, NM 87501, USA\\
   \texttt{http://davidwolpert.weebly.com}}


\begin{abstract}
Recent advances in nonequilibrium statistical physics have provided unprecedented 
insight into the thermodynamics of dynamic processes. The author recently used these advances to
extend Landauer's semi-formal reasoning concerning the thermodynamics of bit erasure, to derive the minimal free energy
required to implement an arbitrary computation. Here, I extend this analysis, deriving the minimal free energy
required by an organism to run a given (stochastic) map $\pi$ from its sensor inputs to its actuator outputs.
I use this result to calculate the input-output map $\pi$ of an organism that optimally trades off the free energy 
needed to run $\pi$ with the phenotypic fitness that results from implementing $\pi$.
I end with a general
discussion of the limits imposed on the rate of the terrestrial biosphere's information processing
by the flux of sunlight \mbox{on the Earth}.
\end{abstract}

\maketitle










\section{Introduction}

It is a truism that biological systems acquire and store information about their environments \cite{frank2009natural,frank2012natural,donaldson2010fitness,krakauer2011darwinian,taylor2007information}. 
However, they do not just store information; they also process that information. In other words, they perform computation.
The energetic consequences for biological systems of these three processes---acquiring, storing, and processing
information---are becoming the focus
of an increasing body of research~\cite{bullmore2012economy,sartori2014thermodynamic,mehta2012energetic,mehta2015landauer,laughlin2001energy,govern2014energy,govern2014optimal,lestas2010fundamental,england2013statistical,landenmark2015estimate}. 
In this paper, I further this research
by analyzing the energetic resources that an organism needs in order to compute in a \mbox{fitness-maximizing way.}

Ever since Landauer's seminal work~\cite{landauer1961irreversibility,landauer1996minimal,landauer1996physical,bennett1973logical,bennett1982thermodynamics,bennett1989time,bennett2003notes,maroney2009generalizing,plenio2001physics,shizume1995heat,fredkin2002conservative},
it has been appreciated that the laws of statistical physics
impose lower bounds on how much thermodynamic work must be done on a system
in order for that system to undergo a two-to-one map, e.g., to undergo bit erasure. 
By conservation of energy, that work must ultimately be acquired from some external source 
(e.g., sunlight, carbohydrates, \textit{etc.}).
If that work on the system is eventually converted into heat that is dumped into an external heat bath, then
the system acts as a heater. In the context of biology, this means that
whenever a biological system (deterministically) undergoes a two-to-one map, it must use free energy from 
an outside source to do so and produces heat as a result.

These early analyses led to a widespread belief that there must
be strictly positive lower bounds on how much free energy is required to implement {any} 
 deterministic, logically-irreversible computation. Indeed, Landauer wrote ``...logical irreversibility is associated with physical 
irreversibility and requires a minimal heat generation''~\cite{landauer1961irreversibility}.
In the context of biology, such bounds would translate to a 
lower limit on how much free energy a biological system must ``harvest'' from its environment in order to implement any
particular (deterministic) computation, not just bit erasure.

A related conclusion of these early analyses was that a one-to-two map, in which noise is added to a system that
is initially in one particular state with probability one,
can act as a {refrigerator} 
, rather than a heater, {removing} heat from the environment~\cite{bennett1989time,bennett2003notes,landauer1961irreversibility,bennett1982thermodynamics}.
Formally, the minimal work that needs to be done on a system in order to make it undergo a one-to-two map is {negative}.
So for example, if the system is coupled to a battery that stores free energy,
a one-to-two map can ``power the battery'', by gaining free energy from a
heat bath rather than dumping it there. To understand this intuitively,
suppose we have a two-state system that is initially in one particular state with probability one. 
Therefore, the system initially has low entropy. That means we can connect it to a heat bath
and then have it do work on a battery (assuming the battery was initially at less than
maximum storage), thereby transferring energy from the heat bath into that battery. 
As it does this, though, the system gets thermalized, \textit{i.e.}, undergoes a one-to-two map
(as a concrete example, this is what happens in adiabatic demagnetization of an 
Ising spin system~\cite{landauer1961irreversibility}).

This possibility of gaining free energy by adding noise to a computation, or at least reducing the amount of
free energy the computation needs, means that there is a trade-off in biology: on the one hand, there is a benefit
to having biological computation that is as precise as possible, in order to maximize the behavioral fitness that results
from that computation; on the other hand, there is a benefit to having 
the computation be as imprecise as possible, in order to minimize the amount of free
energy needed to implement that computation. This tradeoff raises the intriguing possibility that some biological systems 
have noisy dynamics ``on purpose'', as a way to maintain high stores of free energy. For such a system, the noise would not be an
unavoidable difficulty to be overcome, but rather a resource to be exploited.

More recently, there has been dramatic progress in our understanding of
non-equilibrium statistical physics and its relation to information-processing~\cite{faist2012quantitative,touchette2004information, sagawa2009minimal,dillenschneider2010comment,sagawa2012fluctuation,crooks1999entropy,crooks1998nonequilibrium,chejne2013simple, jarzynski1997nonequilibrium,esposito2011second,esposito2010three,parrondo2015thermodynamics,pollard2014second,seifert2012stochastic,takara2010generalization,hasegawa2010generalization,prokopenko2015information}. 
Much of this recent literature 
has analyzed the minimal work required
to drive a physical system's (fine-grained) microstate dynamics during the interval from $t=0$ to $t=1$
in such a way that the associated dynamics over some space of (coarse-grained) macrostates is given by some 
specified Markov kernel $\pi$.
In particular, there has been detailed analysis of the minimal work needed 
when there are only two macrostates, $v =0$ and $v=1$, and we require that both get mapped
by $\pi$ to the macrostate $v=0$~\cite{esposito2011second,sagawa2014thermodynamic,parrondo2015thermodynamics}. By identifying the 
macrostates $v \in V$ as Information Bearing Degrees of Freedom (IBDF)~\cite{bennett2003notes} 
of an information-processing device like a digital computer, these
analyses can be seen as elaborations of the analyses of Landauer \textit{et al}. on the thermodynamics of bit erasure.
Recently, these analyses of maps over binary spaces $V$ have been applied to 
explicitly biological systems,
at least for the special case of a periodic forcing function~\cite{england2013statistical}.

These analyses have resulted in substantial clarifications of Landauer's semiformal reasoning, arguably overturning it in
some regards. For example, this analysis has shown that the logical (ir)reversibility of $\pi$ has nothing to
do with the thermodynamic (ir)reversibility of a system that implements $\pi$. In particular, it is
possible to implement bit erasure (which is logically irreversible) in a thermodynamically-reversible
manner. In the modern understanding, there is no irreversible increase of
entropy in bit erasure. Instead, there is a minimal amount of thermodynamic work that needs
to be expended in a (thermodynamically reversible) implementation of bit erasure (see Example~\ref{ex:landauer} below.)

Many of these previous analyses consider processes for implementing $\pi$ that are tailored
for some specific input distribution over the macrostates, $P(v_t)$. Such processes
are designed to be thermodynamically reversible {when run on $P(v_t)$. However, when run on
a distribution other than $P(v_t)$, they are thermodynamically irreversible,} resulting in
wasted (dissipated) work. For example, in~\cite{mandal2012work}, the amount of work required to implement $\pi$ depends
on an assumption for $\epsilon$, the probability of a one in a randomly-chosen position
on the bit string.

In addition, important as they are, these recent analyses are not applicable to arbitrary maps $\pi$
over a system's macrostates. For
example, as discussed in~\cite{wolpert_landauer_2016a}, the ``quench-based'' devices analyzed 
in~\cite{esposito2011second,sagawa2014thermodynamic,parrondo2015thermodynamics} can
only implement maps whose output is independent of its input (as an example, the output of bit
erasure, an erased bit, is independent of the original state of the bit). 

Similarly, the devices considered in~\cite{mandal2012work,barato2014stochastic}
combine a ``tape'' containing a string of bits with
a ``tape head'' that is positioned above one of the bits on the tape. In each iteration of the system,
the bit currently under the tape head undergoes an arbitrary map to produce a new bit value, and then, the tape
is advanced so that the system is above the next bit. 
Suppose that, inspired by~\cite{deffner2013information}, we 
identify the state of the IBDF of the overall tape-based system as the entire bit string, aligned
so that the current tape position of the read/write subsystem is above Bit {zero}
%
. In
other words, we would identify each state of the IBDF as an aligned big string 
$\{v_i : i = \ldots, -1, 0, ... N\}$ where $N$ is the number of bits that
have already been processed, and the (negative) minimal index could either be finite or
infinite (note that unless we specify which bit of the string is the current one, \textit{i.e.},
which has index {zero},
the update map over the string is not defined).

This tape-based system is severely restricted in the set of computations it can implement on its IBDF. 
For~example, because the tape can only move forward, the
system cannot deterministically map an IBDF state
$v = \{\ldots v_{-1}, v_0, v_1, \ldots, v_N\}$ to an IBDF state
$v' = \{\ldots v'_{-1}, v'_0, v'_1, \ldots, v'_{N-1}\}$. (In~\cite{barato2013autonomous},
the tape can rewind. However, such rewinding only arises due to thermal fluctuations
and therefore does not overcome the problem.)


It should be possible to extend either the quench-based devices reviewed in~\cite{parrondo2015thermodynamics}
and the tape-based device introduced in~\cite{mandal2012work} into a system
that could perform arbitrary computation. In fact, in~\cite{wolpert_landauer_2016a},
I showed how to extend quench-based devices into systems that could perform
arbitrary computation in a purely thermodynamically-reversible manner. This allowed me to calculate the
minimal work that any system needs to implement {any} given conditional distribution $\pi$. To be 
precise, I showed how for any $\pi$ and initial distribution $P(v_t)$, one \mbox{could construct}:

\begin{itemize}
\item a physical system $\mathcal{S}$;
\item a process $\Lambda$ running over $\mathcal{S}$; 
\item an associated coarse-grained set $V$ giving the macrostates of $\mathcal{S}$; 
\end{itemize}
such that: 
\begin{itemize}
\item running $\Lambda$ on $\mathcal{S}$ ensures that the distribution across $V$ changes according
to $\pi$, even if the initial distribution differs from $P(v_t)$;
\item $\Lambda$ is thermodynamically reversible if applied to $P(v_t)$.
\end{itemize}

By the second law, no process can implement $\pi$ on $P(v_t)$ with less work than $\Lambda$ requires. 
Therefore, by calculating the amount of work required by $\Lambda$, we calculate a lower
bound on how much work is required to run $\pi$ on $P(v_t)$. In the context of biological
systems, that bound is the minimal amount of free energy that any organism 
must extract from its external environment in order to run $\pi$.

However, just like in the systems considered previously in the literature, this $\Lambda$ is thermodynamically optimized for
that initial distribution $P(v_t)$. It would be thermodynamically irreversible (and therefore
dissipate work) if used for any other other initial distribution. In the context of biological
systems, this means that while natural selection may produce an information-processing organism that
is thermodynamically optimal in one environment, it cannot produce one that is thermodynamically 
optimal in all environments.

%
%
%

Biological systems are not only information-processing systems, however. As mentioned above,
they also acquire information from their environment and store it. Many of these processes
have nonzero minimal thermodynamic costs, \textit{i.e.}, the
system must acquire some minimal free energy to implement them.
In addition, biological systems often rearrange matter, thereby changing its entropy. Sometimes, these
systems benefit by decreasing entropy, but sometimes, they benefit by increasing entropy, 
e.g., as when cells use depletion forces, when 
they exploit osmotic pressures, \textit{etc}. This is another contribution to their free energy requirements. 
Of~course, biological systems also typically perform
physical ``labor'', \textit{i.e.}, change the expected energy of various systems, by breaking/making chemical bonds, 
and on a larger scale, moving objects (including themselves), developing, growing, \textit{etc}. They must harvest free energy 
from their environment to power this labor, as well. 
Some biological processes even involve several of these 
phenomena simultaneously, e.g., a biochemical pathway that processes information from the environment, 
making and breaking chemical bonds as it does so and also changing its overall entropy.
%

%

In this paper, I analyze some of these contributions to the free energy requirements of biological
systems and the implications of those costs for natural selection. The precise contributions of this paper are:
\begin{enumerate}
\item Motivated by the example of a digital computer, 
the analysis in~\cite{wolpert_landauer_2016a} was formulated for systems that change 
the value $v$ of a single set of physical variables, $V$. Therefore, for example,
as formulated there, bit erasure means a map that sends both $v_t = 0$ and $v_t=1$ to $v_{t+1}=0$.

Here, I instead formulate the analysis for biological ``input-output''
systems that implement an arbitrary stochastic map
taking one set of ``input'' physical variables $X$, representing the state of a sensor,
to a separate set of ``output'' physical variables, $Y$, representing the action taken by the organism in
response to its sensor reading. 
Therefore, as formulated in this paper, ``bit erasure'' means a map $\pi$ that sends both $x_t = 0$ and $x_t=1$ to $y_{t+1}=0$.
My first contribution is to show how to implement any given stochastic map $X \rightarrow Y$ 
with a process that requires minimal work if it is applied to some specified distribution over $X$ and to calculate
that minimal work.
\item 
In light of the free energy costs associated with implementing a map $\pi$, what $\pi$ would
we expect to be favored by natural selection? In particular,
recall that adding noise to a computation can result in a reduction in how much work is needed to implement it. Indeed,
by using a sufficiently noisy $\pi$, an organism can {increase} its stored free energy
(if it started in a state with less than maximal entropy).
Therefore, noise might not just be a hindrance that an organism needs to circumvent; an organism
may actually exploit noise, to ``recharge its battery''.
This implies that an organism will want to implement a ``behavior'' $\pi$ that is noisy as possible. 

In addition, not all terms in a map $x_t \rightarrow y_{t+1}$ are equally important to an organism's reproductive
fitness. It will be important to be very precise in what output is produced for some
inputs $x_t$, but for other inputs, precision is not so important. Indeed, for some
inputs, it may not matter at all what output the organism produces in response.
In light of this, natural selection would be expected to favor organisms that
implement behaviors $\pi$ that are as noisy
as possible (thereby saving on the amount of free energy the organism needs to
acquire from its environment to implement that behavior), 
while still being precise for those inputs where behavioral fitness requires it.
I write down the equations for what $\pi$ optimizes this tradeoff and show
that it is approximated by a Boltzmann distribution over a sum of behavioral fitness and energy.
I then use that Boltzmann distribution to calculate a lower bound
on the maximal reproductive fitness over all possible behaviors $\pi$.

\item My last contribution is to use the preceding results to relate the free energy flux incident
on the entire biosphere to the maximal ``rate of computation'' implemented by 
the biosphere. This relation gives an upper bound on the rate of
computation that humanity as a whole can ever achieve, if it restricts itself
to the surface of Earth.
\end{enumerate}

In Section~\ref{sec:prelims}, I first review some of the basic quantities considered in
nonequilibrium statistical physics and then review some of the relevant recent work
in nonequilibrium statistical physics (involving ``quenching processes'') related to the free energy cost of computation. I then
discuss the limitations in what kind of computations that recent work can be used to
analyze. I end by presenting an extension to that recent work that does not
have these limitations (involving ``guided quenching
%
 processes''). In Section~\ref{sec:organisms}, I use this extension to 
calculate the minimal free energy cost of any given input-output ``organism''. I end
this section by analyzing a toy model of the role that this free energy cost would
play in natural selection. Those interested mainly in these biological implications
can skip Section~\ref{sec:prelims} and should still be able to follow the thrust
of the analysis.

In this paper I extend the construction reviewed in~\cite{parrondo2015thermodynamics}
to show how to construct a system to perform any given computation in a thermodynamically reversible manner. (It seems likely that the tape-based system introduced in~\cite{mandal2012work} could also be extended to do this.)

\section{Formal Preliminaries}\unskip
\label{sec:prelims}

\subsection{General Notation}
\label{sec:gen_not}

I write $|X|$ for the cardinality of any countable space $X$. I will write the
Kronecker delta between any two elements $x, x' \in X$ as $\delta(x, x')$. 
For any logical condition $\zeta$, $I(\zeta) = 1$ (0, {respectively}) 
if $\zeta$ is true (false, {respectively}). When referring
generically to any probability distribution, I will write ``$Pr$''. Given any 
distribution $p$ defined over some space $X$,
I write the Shannon entropy for countable $X$, measured in nats, as:
\ba
S_p(X) &=& -\sum_{x\in X} p(x) \ln\bigg [p(x) \bigg] 
\ea

As shorthand, I sometimes write $S_p(X)$ as $S(p)$ or even just $S(X)$ when $p$ is implicit. 
I use similar notation for conditional entropy, joint entropy of more than one random variable, \textit{etc}. 
I also write mutual information between two random variables $X$ and $Y$ in the usual way, as $I(X; Y)$~\cite{mack03,coth91,yeung2012first}.

Given a distribution $q(x)$ and a conditional distribution $\pi(x' \mid x)$,
I will use matrix notation to define the distribution $\pi q$:
\ba
[\pi q](x') &=& \sum_{x} {\pi}(x' \mid x) q(x)
\ea

For any function $F(x)$ and distribution $P(x)$, I write:
\ba
\E_P(F) &=& \sum_x F(x) P(x)
\ea

I will also sometimes use capital letters to indicate variables that are marginalized
over, e.g., writing:
\ba
\E_P(F(X, y)) &=& \sum_x P(x) F(x, y)
\ea

Below, I often refer to a process as ``semi-static''. This means that these processes transform one Hamiltonian
into another one so slowly that the associated distribution is
always close to equilibrium, and as a result, only infinitesimal amounts of dissipation
occur during the entire process. For this assumption to be valid, the implicit units of time 
in the analysis below must be sufficiently long on the timescale of the 
relaxation processes of the physical systems involved (or equivalently, those relaxation
processes must be sufficiently quick when measured in those time units).

If a system with states $x$ is subject to a Hamiltonian $H(x)$, then the
associated equilibrium free energy is:
\ba
F_{eq}(H) &\equiv& -\beta^{-1} \ln[Z_H(\beta)]
\ea
where as usual $\beta \equiv 1 / k T$, and the partition function is:
\ba
Z_H(\beta) &=& \sum_x \exp{-\beta H(x)}
\ea

However, the analysis below focuses on nonequilibrium distributions $p(x)$, for
which the more directly relevant quantity is the {nonequilibrium} free energy,
in which the distribution need not be a Boltzmann distribution for the current Hamiltonian:
\ba
F_{neq}(H, p) &\equiv& \E_p(X) - kTS(p) \nonumber \\
 &=& \sum_x p(x) H(x) + kT\sum_x p(x) \ln[p(x)]
\ea
where $k$ is Boltzmann's constant.
For fixed $H$ and $T$, $F_{neq}(H, p)$ is minimized by the associated Boltzmann distribution $p$, for
which it has the value $F_{eq}(H)$. It will be useful below to consider the changes in nonequilibrium free
energy that accompany a change from a distribution $P$ to a distribution $M$ accompanied by
a change from a Hamiltonian $H$ to a Hamiltonian $H'$:
\ba
\Delta F_{neq}^{H,H'}(P, M) &\equiv& F_{neq}(H', M) - F_{neq}(H, P)
\label{eq:delta_noneq_free_def}
\ea

\subsection{Thermodynamically-Optimal Processes}
\label{sec:2.2}

If a process $\Lambda$ maps a distribution $P$ to a distribution $M$ thermodynamically reversibly,
then the amount of work it uses when applied to $P$ is 
$\Delta F_{neq}^{H,H'}(P, M)$~\cite{parrondo2015thermodynamics,reif65,still2012thermodynamics,deffner2013information}. In particular,
$\Delta F_{neq}^{H,H'}(P, \pi P)$ is the amount of work used by a thermodynamically-reversible
process $\Lambda$ that maps a distribution $P$ to $\pi P$.
Equivalently, it is negative for the amount of work that is extracted by $\Lambda$ when
transforming $P$ to $\pi P$. 

In addition, by the second law, there is no process that
maps $P$ to $M$ while requiring less work than a thermodynamically-reversible process that maps $P$ to $M$.
This motivates the following~definition.

\begin{definition}
Suppose a system undergoes a process $\Lambda$ that starts with Hamiltonian 
$H$ and ends with Hamiltonian $H'$. Suppose as well that: 
\begin{enumerate}
\item at both the start and finish of $\Lambda$,
the system is in contact with a (single) heat bath at temperature~$T$;
\item $\Lambda$ transforms any starting distribution $P$ to an ending distribution $\pi P$, where neither of
those two distributions need be at equilibrium for their respective Hamiltonians;
\item $\Lambda$ is thermodynamically reversible when run on some particular starting distribution $P$.
\end{enumerate}
Then, $\Lambda$ is \textit{thermodynamically optimal} for the tuple $(P, \pi , H, H')$.
\end{definition} 

\begin{example}
Suppose we run a process over a space $X \times Y$, transforming the $t=0$ 
distribution $q(x)M(y)$ to a $t=1$ distribution $p(x)M(y)$. Therefore, $x$ and $y$ are statistically
independent at both the beginning and the end of the process, and
while the distribution over $x$ undergoes a transition from $q \rightarrow p$, 
the distribution over $y$ undergoes a cyclic process, taking $M \rightarrow M$ (note that
it is {not} assumed that the ending and starting $y$'s are the same or that $x$ and $y$
are independent at times between $t=0$ and $t=1$).

Suppose further that at
both the beginning and end of the process, there is no interaction Hamiltonian, \mbox{\textit{i.e.}, at those two times}:
\ba
H(x, y) &=& H^X(x) + H^Y(y)
\ea

Then, no matter how $x$ and $y$ are coupled during the process, no matter how smart
the designer of the process, the process will require work of at least:
\ba
\Delta F_{neq}^{H,H}(q, p) &=& \bigg(E_p(H^X) - E_q(H^X)\bigg) - kT \bigg(S(p) - S(q) \bigg)
\label{eq:lower_bound_example}
\ea

Note that this amount of work is independent of $M$.
\end{example}

As a cautionary note, the work expended by any process operating on
any initial distribution $p(x)$ is the average of the work expended on each
$x$. However, the associated change in nonequilibrium free energy is not the average of the change in nonequilibrium
free energy for each $x$. This is illustrated in the following example.

\begin{example}
Suppose we have a process $\Lambda$ that sends each initial $x$ to an associated
final distribution $\pi(x' \mid x)$, while transforming the initial Hamiltonian $H$ into the final Hamiltonian $H'$. 
Write $W^\Lambda_{H, H',\pi}(x)$ for the work expended by $\Lambda$ 
when it operates on the initial state $x$. Then, the work expended by $\Lambda$ 
operating on an initial distribution $p(x)$ is $\sum_x p(x) W^\Lambda_{H, H',\pi}(x)$. In particular,
choose the process $\Lambda$, so that it sends $p \rightarrow \pi p$ with minimal work. Then:
\ba
\sum_x p(x) W^\Lambda_{H, H',\pi}(x) &=& \Delta F^{H',H}_{neq}(p, \pi p)
\ea

However, 
this does \emph{not} equal the average over $x$ of the 
associated changes to nonequilibrium free \mbox{energy, \textit{i.e.},}
\ba
 \Delta F^{H',H}_{neq}(p, \pi p) &=& F_{neq}(H', \pi p) - F_{neq}(H, p) \nonumber \\
 &\neq& \sum_x p(x) \bigg[F_{neq}(H', \pi(Y \mid x)) - F_{neq}(H, \delta(X,x)) \bigg] \nonumber \\
\ea
(where $\delta(X,x)$ is the distribution over $X$ that is a delta function at $x$).
The reason is that the entropy terms in those two nonequilibrium free energies
are not linear; in general, for any probability distribution $Pr(x)$,

\ba
\sum_x Pr(x) \ln[Pr(x)] &\ne& \sum_x Pr(x) \sum_{x'} \delta(x', x) \log[\delta(x', x)]
\ea
\end{example}

I now summarize what will be presented in the rest of this section.

Previous work showed how to construct a thermodynamically-optimal process
for many tuples $(p, \pi , H, H')$. 
In particular, as discussed in the Introduction,
it is known how to construct a thermodynamically-optimal process for
any tuple $(p, \pi , H, H')$ where $\pi(x'\mid x)$ is independent of $x$, like
bit erasure. Accordingly, we know the minimal work necessary 
to run any such tuple. In Section~\ref{sec:q_proc_def}, I review
this previous analysis and show how to apply it to the kinds of input-output systems considered
in this paper.

However, as discussed in the Introduction, until recently, it was not known whether
one could construct a thermodynamically-optimal process for {any} tuple
$(p, \pi , H, H')$. In particular, given an arbitrary pair of an initial distribution $p$
and conditional distribution $\pi$, it was not known whether there is a process
$\Lambda$ that is thermodynamically optimal for $(p, \pi , H, H')$ for some
$H$ and $H'$. This means that it was not known what the minimal
needed work is to apply an arbitrary stochastic map $\pi$ to
an arbitrary initial distribution $p$. In particular, it was not known
if we could use the difference in nonequilibrium free energy between $p$ and $\pi p$
to calculate the minimal work needed to apply a computation $\pi$ to an
initial distribution $p$.

This shortcoming was overcome in~\cite{wolpert_landauer_2016a}, where it was explicitly shown how
to construct a thermodynamically-optimal process for any tuple
$(p, \pi , H, H')$. In Section~\ref{sec:gq_proc_def}, I show in detail how to
construct such processes for any input-output system.

Section~\ref{sec:gq_proc_def} also discusses the fact that a process that
is thermodynamically optimal for $(p, \pi , H, H')$ need {not} be thermodynamically
optimal for $(p', \pi , H, H')$ if $p' \ne p$. Intuitively, if we construct a process
$\Lambda$ that results in minimal required work for initial distribution $p$ and
conditional distribution $\pi$, but then apply that machine to a different distribution
$p' \ne p$, then in general, work is dissipated. While that $\Lambda$ is thermodynamically
reversible when applied to $p$, in general, it is not thermodynamically
reversible when applied to $p' \ne p$. As an example, if we design a computer
to be thermodynamically reversible for input distribution $p$, but
then use it with a different distribution of inputs, then work is dissipated.

In a biological context, this means that if an organism is ``designed'' not to dissipate any work when
it operates in an environment that produces inputs according to some $p$, but
instead finds itself operating in an environment that produces inputs according to
some $p' \ne p$, then it will dissipate extra work. That dissipated work is wasted since it
does not change $\pi$, \textit{i.e.}, has no consequences for the input-output map
that the organism implements. However, by the conservation of energy, that
dissipated work must still be acquired from some external source. This means that
the organism will need to harvest free energy from its environment at a higher rate \mbox{(to 
supply that dissipated work)} than would an organism that were ``designed''
for $p'$.

\subsection{Quenching Processes}
\label{sec:q_proc_def}

A special kind of process, often used in the literature, can be used to transform any given initial nonequilibrium distribution
into another given nonequilibrium distribution in a thermodynamically-reversible manner. These processes
begin by quenching the Hamiltonian of a system. After that, the Hamiltonian
is isothermally and quasi-statically changed, with the system in continual contact with a heat bath at a fixed
temperature $T$. The process ends by applying a reverse quench to return to the original 
Hamiltonian \mbox{(see~\cite{parrondo2015thermodynamics,esposito2011second,sagawa2014thermodynamic}}
for discussion of these kinds of processes).

More precisely, such a \emph{Quenching (Q)
 process} applied to a system
with microstates $r \in R$ is defined by:
\begin{enumerate}
\item an \textit{initial}/\textit{final} Hamiltonian $H^t_{sys}(r)$;
\item an \textit{initial} distribution $\rho^t(r)$;
\item a \textit{final} distribution $\rho^{t+1}(r)$;
\end{enumerate}
and involves the following three steps:

\begin{enumerate}
\item[(i)] To begin, the system has Hamiltonian $H^t_{sys}(r)$, which is quenched
into a first \textit{quenching Hamiltonian}: 
\ba
H^t_{quench}(r) &\equiv& -kT \ln[\rho^t(r)]
\label{eq:first_quench}
\ea
In other words, the Hamiltonian is changed from $H^t_{sys}$ to
$H^t_{quench}$ too quickly for the distribution over $r$ to change from $\rho^t(r)$.

%

Because the quench is effectively instantaneous, it is thermodynamically
reversible and is adiabatic, involving no heat transfer between the system and the
heat bath. On the other hand, while $r$ is unchanged in a quench and, therefore, so is the distribution over $R$,
in general, work is required if $H^t_{quench} \ne H^t_{sys}$
(see~\cite{crooks1999entropy,crooks1998nonequilibrium,still2012thermodynamics,reif65}).

Note that {if} the Q process is applied to the distribution $\rho^t$, then at the end
of this first step, the distribution is at thermodynamic equilibrium. However, if the process is
applied to any other distribution, this will not be the case. In this situation, work is unavoidably dissipated in
in the next step.

\item[(ii)]
Next, we isothermally and quasi-statically transform $H^t_{quench}$ to a second quenching Hamiltonian, 
\ba
H^{t+1}_{quench}(r) &\equiv& -kT \ln[\rho^{t+1}(r)]
\label{eq:second_quench}
\ea
Physically, this means two things. First, that
a smooth sequence of Hamiltonians, starting with $H^t_{quench}$ and ending with
$H^{t+1}_{quench}$, is applied to the system. Second, that while that sequence is being
applied, the system is coupled with an external heat bath at temperature $T$, where the
relaxation timescales of that coupling are arbitrarily small on the time scale of the dynamics
of the Hamiltonian. This second requirement ensures that to first order, the system is always in thermal equilibrium for
the current Hamiltonian, assuming it started in equilibrium at the beginning of the step
(recall from Section~\ref{sec:gen_not} that I assume that quasi-static transformations occur in an arbitrarily small amount of time,
since the relaxation timescales are \mbox{arbitrarily short)}.

\item[(iii)] Next, we run a quench over $R$ ``in reverse'', instantaneously
replacing the Hamiltonian $H^{t+1}_{quench}(r)$ with the initial Hamiltonian $H^t_{sys}$,
with no change to $r$. As in step (i), while work may be done (or extracted) in
step (iii), no heat is transferred.
\end{enumerate}

Note that we can specify any Q process in terms of its first and second quenching Hamiltonians rather than in terms
of the initial and final distributions, since there is a bijection between those two pairs. This central
role of the \emph{q}uenching Hamiltonians is the basis of the name ``Q'' process 
(I distinguish the distribution $\rho$ that defines a Q process, which is instantiated in the physical 
structure of a real system, from the actual distribution $P$ on which that physical system is run).

Both the first and third steps of any Q process are thermodynamically reversible, no matter
what distribution that process is applied to. In addition, 
if the Q process is applied to $\rho^t$, the second step will be thermodynamically reversible.
Therefore, as discussed 
in~\cite{parrondo2015thermodynamics,still2012thermodynamics,deffner2013information, esposito2011second},
if the Q process is applied to $\rho^t$, then the expected work expended by the process is given by the 
change in nonequilibrium free energy in going from $\rho^t(r)$ to $\rho^{t+1}(r)$,
\ba
&&\Delta F_{neq}^{H^t_{sys},H^t_{sys}}(\rho^t, \rho^{t+1}) \nonumber \\
&& \qquad = \;  \E_{\rho^{t+1}}(H^t_{sys}) - \E_{\rho^t}(H^t_{sys}) + kT\bigg[S(\rho^t) - S(\rho^{t+1})\bigg] 
\label{eq:change_v}
\ea

Note that because of how $H^t_{quench}$ and $H^{t+1}_{quench}$ are defined,
there is no change in the nonequilibrium free energy during the second step of the Q process
if it is applied to $\rho^t$:
\ba
\E_{\rho^{t+1}}(H^{t+1}_{quench}) - \E_{\rho^t}(H^t_{quench}) + kT\bigg[S(\rho^t) - S(\rho^{t+1})\bigg] &=& 0 \nonumber \\
\ea
All of the work arises in the first and third steps, involving the two quenches.

The relation between Q processes and information-processing of macrostates arises
once we specify a partition over $R$. I end this subsection with the following example of a Q process:

\begin{example}
\label{ex:landauer}
Suppose that $R$ is partitioned into two bins, \textit{i.e.}, there are two macrostates. For both $t = 0$ and $t=1$, for both partition
elements $v$, with abuse of notation, define:
\ba
P^t(v) &\equiv& \sum_{r \in v} \rho^t(r \mid v)
\ea
so that:
\ba
\rho^t(r) &=& \sum_v P^t(v) \rho^t(r \mid v)
\ea

Consider the case where $P^0(v)$ has full support, but $P^1(v) = \delta(v,0)$. Therefore,
the dynamics over the macrostates (bins) from $t=0$ to
$t=1$ sends both $v$'s to zero. In other words, it erases a bit. 

For pedagogical simplicity, take $H^0_{sys} = H^1_{sys}$ to be uniform.
Then, plugging in to~Equation (\ref{eq:change_v}), we see that the minimal work is:
\ba
kT[S(\rho^0) - S(\rho^1)] &=& kT\bigg[S(P^0) + \sum_{v} P^0(v) \bigg(-\sum_{r} P^0(r \mid v) \ln[\rho(r \mid v)] \bigg) \bigg] \nonumber \\
		&& \qquad \;-\; \{0 \rightarrow 1\} \nonumber \\
  &=& kT\bigg[S(P^0) + \sum_{v^0} P^0(v) S(R^0 \mid v^0)\bigg] \;-\; \{0 \rightarrow 1\} \nonumber \\
 &=& kT\bigg[S(P^0) + S(R^0 \mid V^0) - S(P^1) - S(R^1 \mid V^1) \bigg] \nonumber \\
 &=& kT\bigg[S(P^0) + S(R^0 \mid V^0) - S(R^1 \mid V^1) \bigg]  
\label{eq:generalized_landauer_bound} 
\ea
(the two terms $S(R^t \mid v^t)$ are sometimes called ``internal entropies'' in the literature~\cite{parrondo2015thermodynamics}).

In the special case that $P^0(v)$ is uniform and that $S(R^t \mid v^t)$ is the same for
both $t$ and both $v_t$, we recover Landauer's bound, $kT\ln(2)$, as the minimal amount
of work needed to erase the bit.
Note though that outside of that special case, Landauer's bound does not 
give the minimal amount of work needed to erase a bit. Moreover,
in all cases, the limit in~Equation (\ref{eq:generalized_landauer_bound}) is on the amount of work needed to erase the bit; 
a bit can be erased with zero dissipated work, \emph{pace} Landauer.
For this reason, the bound in~Equation (\ref{eq:generalized_landauer_bound}) is sometimes
called ``generalized Landauer cost'' in the literature~\cite{parrondo2015thermodynamics}.

On the other hand, suppose that we build a device to implement a Q process that achieves the bound
in~Equation (\ref{eq:generalized_landauer_bound}) for one particular initial distribution over the value of the bit, $\G_0(v)$.
Therefore, in particular, that device has ``built into it'' a first and second quenching Hamiltonian
given by:
\ba
H^0_{quench}(r) &=& -kT\ln[\G_0(r)] \\
H^1_{quench}(r) &=& -kT \ln[\G_1(r)]
\ea
respectively, where:
\ba
\G_0(r) &\equiv& \sum_v \G_0(v) \rho^0(r \mid v) \\
\G_1(r) &\equiv& \rho^1(r \mid v = 0)
\label{eq:G_1_def}
\ea

 If we then apply that device with a different initial macrostate distribution, 
$\PP_1(v) \ne \G_0(v)$, in general, work will be dissipated
in step (ii) of the Q process, because $\PP_1(r) = \sum_v \PP_1(v) \rho^0(r \mid v)$
will not be an equilibrium for $H^0_{quench}$. In the
context of biology, if a bit-erasing organism is optimized for one environment, but then used in
a different one, it will necessarily be inefficient, dissipating work (the minimal amount of work
dissipated is given by the drop in the value of
the Kullback--Leibler divergence between $\G_t$ and $\PP_t$ as the system develops from $t=0$ to $t=1$;
see~\cite{wolpert_landauer_2016a}).
\end{example}

\subsection{Guided Q Processes}
\label{sec:gq_proc_def}

Soon after the quasi-static transformation step of any Q process begins,
the system is thermally relaxed. Therefore, all information about $r_t$, the initial value
of the system's microstate, is quickly removed from the distribution over $r$ (phrased differently, 
that information has been transferred into inaccessible degrees of freedom in the external heat bath).
This means that the second quenching Hamiltonian cannot depend on the 
initial value of the system's microstate; after that thermal relaxation
of the system's microstate, there is no degree of freedom in the microstate that has
any information concerning the initial microstate. This means that after the
relaxation, there is no degree of freedom within the system undergoing the Q process that can modify the second 
quenching Hamiltonian based on the value of the initial microstate. 

As a result, \emph{by itself}, a Q process cannot change 
an initial distribution in a way that depends on that initial distribution. In particular,
it cannot map different initial macrostates to different final macrostates
(formally, a Q process cannot map a distribution with support restricted to the microstates in the macrostate
$v_t$ to one final distribution and map a distribution with support restricted to the macrostate
$v'_t \ne v_t$ to a different final distribution). 

On the other hand, both quenching Hamiltonians of a Q process running on a system $\mathcal{R}$ with microstates $r \in R$ {can}
depend on $s_t \in S$, the initial microstate of a different system, $\mathcal{S}$. Loosely speaking,
we can run a process over the joint system $\mathcal{R} \times \mathcal{S}$ that is thermodynamically reversible and
whose effect is to implement a different Q process over $R$, depending on the value $s_t$. 
In particular, we can ``coarse-grain'' such dependence on $s_t$:
given any partition over $\mathcal{S}$ whose elements are labeled by
$v \in V$, it is possible that both quenching Hamiltonians of a Q process running on $\mathcal{R}$
are determined by the macrostate $v_t$.

More precisely, 
a \emph{Guided Quenching (GQ) process} over $R$ guided by $V$ (for conditional distribution ${\overline{\pi}}$ and initial distribution $\rho^t(r, s)$)''
 %
is defined by a quadruple:
\begin{enumerate}
\item an \textit{initial}/\textit{final} Hamiltonian $H^t_{sys}(r, s)$;
\item an \textit{initial} joint distribution $\rho^t(r, s)$;
\item a time-independent partition of $S$ specifying an associated set of macrostates, $v \in V$;
\item a conditional distribution ${\overline{\pi}}(r \mid v)$.
\end{enumerate}
It is assumed that for any $s, s'$ where $s \in V(s')$,
\ba
\rho^t(r \mid s) &=& \rho^t(r \mid s')
\label{eq:s_dont_matter}
\ea
\textit{i.e.}, that the distribution over $r$ at the initial time $t$ can depend on the macrostate $v$, but not on
the specific microstate $s$ within the macrostate $v$. It is also assumed that there are boundary points in $S$
(``potential barriers'') separating the members of $V$ in that the system cannot physically move from $v$ to $v' \ne v$ without
going through such a boundary point. 

The associated GQ process involves the following steps:

\begin{enumerate}
\item[(i)] To begin, the system has Hamiltonian $H^t_{sys}(r, s)$, which is quenched
into a first {quenching Hamiltonian} written as:
\ba
H^t_{quench}(r, s) &\equiv& H^t_{quench;S}(s) + H^t_{quench; int}(r, s)
\ea
We take:
\ba
H^t_{quench; int}(r, s) &\equiv& -kT \ln[\rho^t(r \mid s)]
\ea
and for all $s$ except those at the boundaries of the partition elements defining the
macrostates $V$,
\ba
H^t_{quench;S}(s) &\equiv& -kT \ln[\rho^t(s)]
\label{eq:28}
\ea
However, at the $s$ lying on the boundaries of the partition elements defining $V$,
$H^t_{quench;S}(s)$ is arbitrarily large. Therefore,
there are infinite potential barriers separating the macrostates of $\mathcal{S}$. 

Note that away from those boundaries of the partition elements defining $V$,
$\rho^t(r, s)$ is the equilibrium distribution for $H^t_{quench}$.

\item[(iii)]
Next, we isothermally and quasi-statically transform $H^t_{quench}$ to a second quenching Hamiltonian, 
\ba
H^{t+1}_{quench;S}(r, s) &\equiv& H^t_{quench;S}(s) + H^{t+1}_{quench; int}(r, s)
\label{eq:29}
\ea
where:
\ba
H^t_{quench; int}(r, s) &\equiv& -kT \ln[{\overline{\pi}}(r \mid V(s))]
\ea
($V(s)$ being the partition element that contains $s$).

Note that the term in the Hamiltonian that only concerns $\mathcal{S}$ does not change in this step.
Therefore, the infinite potential barriers delineating partition boundaries in $S$ remain for the entire step.
I~assume that as a result of those barriers, the coupling of $\mathcal{S}$ with the heat bath
during this step cannot change the value of~$v$. As a result, even though the distribution over $r$ changes in this step, there is no change to the value of $v$.
To describe this, I say that $v$ is ``semi-stable'' during this step. (To state this assumption more
formally, let $A(s', s'')$ be the (matrix) kernel
that specifies the rate at which $s' \rightarrow s''$ due to heat transfer between $\mathcal{S}$ and the heat bath 
during during this step (ii)~\cite{crooks1999entropy,crooks1998nonequilibrium}. 
Then, I assume that $A(s', s'')$ is arbitrarily small if $V(s'') \ne V(s')$.)

As an example, the different
bit strings that can be stored in a flash drive all have the same expected energy, but the
energy barriers separating them ensure that the distribution over bit strings relaxes to the uniform
distribution infinitesimally slowly. Therefore, the value of the bit string is semi-stable.

Note that even though a semi-stable system is not at thermodynamic
equilibrium during its ``dynamics'' (in which its macrostate does not change), that dynamics is thermodynamically reversible,
in that we can run it backwards in time without requiring any work or resulting in
heat dissipation.

%
%
%

\item[(iii)] Next, we run a quench over $R \times S$ ``in reverse'', instantaneously
replacing the Hamiltonian $H^{t+1}_{quench}(r,s)$ with the initial Hamiltonian $H^t_{sys}(r,s)$,
with no change to $r$ or $s$. As in step (i), while work may be done (or extracted) in
step (iii), no heat is transferred.
\end{enumerate}



There are two crucial features of GQ processes. The first is that a GQ process faithfully implements ${\overline{\pi}}$ even
if its output varies with its input and does so no matter what the initial distribution over $R \times S$ is.
The second is that for a {particular} initial distribution over $R \times S$, implicitly specified by $H^t_{quench}(r, s)$,
the GQ process is thermodynamically reversible.
 
The first of these features is formalized with the following result, proven in Appendix  A:
\begin{proposition}
A GQ process over $R$ guided by $V$ 
(for conditional distribution ${\overline{\pi}}$ and initial distribution $\rho^t(r, s)$) will transform
any initial distribution $p^t(v) p^t(r \mid v)$ into a distribution $p^t(v) {\overline{\pi}}(r \mid v)$
without changing the distribution over $s$ conditioned on $v$.
\label{prop:22}
\end{proposition}


Consider the special case where the GQ process is in fact applied to the initial distribution that defines it,
\ba
\rho^t(r, s) &=& \sum_v \rho^t(v) \rho^t(s \mid v) \rho^t(r \mid v)
\label{eq:before_pi}
\ea
(recall~Equation (\ref{eq:s_dont_matter})). In this case, the initial distribution is a Boltzmann distribution for the first quenching Hamiltonian; 
the final distribution is:
\ba
\rho^{t+1}(r, s) &=& \sum_v \rho^t(v) \rho^t(s \mid v) {\overline{\pi}}(r \mid v)
\label{eq:after_pi}
\ea
and the entire GQ process is thermodynamically reversible. This establishes the second crucial feature of \mbox{GQ processes.}

Plugging in, in this special case, the change in nonequilibrium free energy is:
\ba
&& \Delta F_{neq}^{H^t_{sys}, H^t_{sys}} (\rho^t, \rho^{t+1}) \nonumber \\
&& \qquad	 = \; \bigg[\sum_{r,s,v} \rho^t(v) \rho^t(s \mid v) \big({\overline{\pi}}(r \mid v) \nonumber \\
&& \qquad \qquad - \rho^t(r \mid v)\big) H^t_{sys}(r, s) \bigg]
		 - kT\bigg[S(\rho^{t+1}) - S(\rho^t)\bigg]
\label{eq:gq_work_required}		 
\ea	

This is the minimal amount of free energy needed to implement the GQ process.	 
An important example of such a thermodynamically-optimal GQ process is the work-free 
copy process discussed in~\cite{parrondo2015thermodynamics} and the \mbox{references therein.}

Suppose that we build a device to implement a GQ process
over $R$ guided by $V$ for conditional distribution ${\overline{\pi}}$ and initial distribution: 
\ba
\rho^t(r, s) &=& \sum_v \rho^t(r \mid v)\rho^t(s \mid v)\G_t(v)
\ea

Therefore, that device has ``built into it'' first and second quenching Hamiltonians
that depend on \linebreak $\rho^t(r \mid v), \rho^t(s \mid v)$ and $\G_t$.
Suppose we apply that device in a situation where the initial distribution over $r$ conditioned on $v$
is in fact $\rho^t(r \mid v)$ and the initial distribution over $s$ conditioned on $v$ is
in fact $\rho^t(s \mid v)$, but the initial macrostate distribution, 
$P_t(v)$, does not equal $\G_t(v)$. In this situation, the actual initial distribution
at the start of step (ii) of the GQ process will not be an equilibrium for the initial quenching Hamiltonian. 
However, this will not result in there being any work dissipated during the thermal relaxation of
that step. That is because the 
distribution over $v$ in that step does not relax, no matter what it is initially
(due to the infinite potential barriers in $S$), while the initial distribution over
$(r, s)$ conditioned on $v$ {is} in thermal equilibrium for the initial quenching Hamiltonian.

However, now suppose that we apply the device in a situation where the initial distribution over $r$ conditioned on $v$
does not equal $\rho^t(r \mid v)$. In this situation, work {will} be dissipated in step (ii) of the GQ process. That is
because the initial distribution over $r$ when the relaxation starts is not in
thermal equilibrium for the initial quenching Hamiltonian, and this distribution {does} relax in step (ii).
Therefore, if the device was not ``designed'' for the actual initial distribution over $r$ conditioned on $v$
(\textit{i.e.}, does not use a $\rho^t(r \mid v)$ that equals that actual distribution), it will necessarily
dissipate~work. 

As elaborated below, this means that if a biological organism 
that implements {any} map ${\overline{\pi}}$ is optimized for one environment, \textit{i.e.}, one distribution over its inputs,
but then used in an environment with a different distribution over its inputs, 
it will necessarily be inefficient, dissipating work (recall that above, we established
a similar result for the specific type of Q process that can be used to erase~a~bit).

%

\section{Organisms} 
\label{sec:organisms}

In this section, I consider biological systems that process an input into an output, an output that
specifies some action that is then taken back to the environment.
As shorthand, I will refer to any biological system that does this as an ``organism''.
A cell exhibiting chemotaxis is an example of an organism, with its input being (sensor
readings of) chemical concentrations and its output being chemical signals that in turn specify some directed motion it will follow.
Another example is a eusocial insect colony, with its inputs
being the many different materials that are brought into the nest (including 
atmospheric gases) and its output being material waste products (including heat) that
in turn get transported out of the colony.


Physically, each organism contains an ``input subsystem'', a ``processor subsystem'' and an ``output subsystem''
(among others). The initial macrostate of the input subsystem is formed by sampling some distribution
specified by the environment and is then copied to the macrostate of the processor subsystem.
Next, the processor iterates some specified first-order 
time-homogenous Markov chain (for example, if the organism is a cell, this Markov chain models
the iterative biochemical processing of the input that takes place within the organism).
The ending value of the chain is the organism's output, which specifies the action that the organism
then takes back to its environment. In general, it could be that for certain inputs, an 
organism never takes any action back to its environment, but instead keeps
processing the input indefinitely. Here, that is captured by having the Markov chain
keep iterating (e.g., the biochemical processing keeps going)
until it produces a value that falls within a certain predefined \textit{halting} (sub)set, which is then 
copied to the organism's output
\mbox{(the possibility that the processing} never halts also ensures that the organism is Turing 
complete~\cite{hopcroft2000jd,livi08,grunwald2004shannon}). 

There are {many}
features of information processing in real biological systems that are distorted in this model; it is
just a starting point. Indeed, some features are absent entirely. In particular,
since the processing is modeled as a {first}-order Markov chain, there is no way for an
organism described by this model to ``remember'' a previous input it received when determining what action
to take in response to a current input. Such features could be incorporated into
the model in a straight-forward way and are the subject of future work.

In the next subsection, I formalize this model of a biological input-output system, in terms of an
input distribution, a Markov transition matrix and a halting set.
I then analyze the minimal amount of work needed by {any} physical system that implements 
a given transition matrix when receiving inputs from a given distribution, \textit{i.e.}, the
minimal amount of work a real organism would need to implement 
its input-output behavior that it exhibits in its environment, if it were free to 
use any physical process that obeys the laws of physics. 
To perform this analysis, I will construct a specific physical process that 
implements an iteration of the Markov transition matrix of a given organism with minimal work, when inputs are
generated according to the associated input distribution.
This process involves a sequence of multiple GQ processes.
\emph{It cannot be emphasized enough that these processes I construct
are $not$ intended to describe what happens in real biological input-output systems, even as a cartoon}. 
These processes are used only as a calculational tool, for finding a lower bound on 
the amount of work needed 
by a real biological organism to implement a given input-output transition matrix.

Indeed, because real biological systems are often quite inefficient, in practice, they will often use
far more work than is given by the bound I calculate. However, we might expect
that in many situations, the work expended by a real biological system that behaves according
to some transition matrix is approximately proportional to 
the work that would be expended by a perfectly efficient system obeying the same
transition matrix. Under that approximation,
the relative sizes of the bounds given below should reflect the relative sizes of the
amounts of work expended by real biological~systems.

\subsection{The Input and Output Spaces of an Organism}

Recall from Section~\ref{sec:gq_proc_def} that a subsystem $S$ cannot use a thermodynamically-reversible Q process
to update its own macrostate in an arbitrary way. However a different subsystem $S'$
can guide an arbitrary updating of the macrostate of $S$, with a GQ process. In addition,
the work required by a thermodynamically-reversible process that implements a given
conditional distribution from inputs to outputs is the same as the work required
by any other thermodynamically-reversible process that implements that same distribution. 

In light of these two facts, for simplicity,
I will not try to construct a thermodynamically-reversible process that implements any given
organism's input-output distribution directly, by iteratively updating the processor until
its state lies in the halting subset and then copying that state to the output. Instead,
I will construct a thermodynamically-reversible process that implements that same input-output distribution, but by
``ping-ponging'' GQ processes back and forth between the
state of the processor and the state of the output system, until the output's state lies in the halting set.

Let $W$ be the space of all possible microstates of a \textit{processor} subsystem, and
$U$ the (disjoint) space of all possible microstates of an \textit{output} subsystem.
Let $\X$ be a partition of $W$, \textit{i.e.}, a coarse-graining of it into a countable set of macrostates.
Let $X$ be the set of labels of those partition elements, \textit{i.e.}, the range of the map $\X$
(for example, in a digital computer, $\X$ could be a map taking each microstate of the computer's main RAM,
$w \in W$, into the associated bit string, ${\cal{X}}(w) \in X$).
Similarly, let $\Y$ be a partition of $U$, the microstate of the output subsystem.
Let $Y$ be the set of labels of those partition elements, \textit{i.e.}, the range of the map $\Y$,
with $Y_{halt} \subseteq Y$ the halting subset of $Y$.
I generically write an element of $X$ as $x$ and an element of $Y$ as $y$.
I assume that $X$ and $Y$, the spaces of labels of the processor and output partition elements, 
respectively, have the same cardinality and, so, 
indicate their elements with the same labels. In particular, if we are concerned with Turing-complete
organisms, $X$ and $Y$ would both be $\{0, 1\}^*$, the set of all finite bit strings (a set that is \mbox{bijective
with $\N$). }

For notational convenience, I arbitrarily choose one non-empty element of 
$X$ and one non-empty element
of $Y$ and the additional label $0$ to both of them (for example, in a Turing machine, it could
be that we assign the label $0$ to the partition element that also has label $\{0\}$). Intuitively,
these elements represent the ``initialized'' state of the processor and output subsystems, respectively.

The biological system also contains an \textit{input} subsystem, with microstates $f \in F$ and
coarse-graining partition $\cal{F}$ that produces macrostates $b \in B$.
The space $B$ is the same as the space $X$ (and therefore is the same as $Y$).
The state of the input at time $t=0$, $b_0$, is formed by sampling an \textit{environment distribution} $\PP_1$. As an example, 
$b_0$ could be determined by a (possibly noisy) sensor reading of the external environment. 
As another example, the environment of an organism could directly perturb 
the organism's input macrostate at $t=0$. 
For~simplicity, I assume that both the processor subsystem and 
the output subsystem are initialized before $b_0$ is generated, \textit{i.e.}, that $x_0 = y_0 = 0$. 

After $b_0$ is set this way, it is copied to the processor subsystem, setting $x_1$. At this
point, we iterate a sequence of GQ processes in which $x$ is mapped to $y$, then
$y$ is mapped to $x$, then that new $x$ is mapped to a new $y$, \textit{etc}., until (and if) $y \in Y_{halt}$.
To make this precise, 
adopt the notation that $[\alpha, \alpha']$ refers to the joint state $(x = \alpha, y = \alpha')$. 
Then, after $x_1$ is set, we iterate the following multi-stage \textit{ping-pong} sequence:
\begin{enumerate}
\item $[x_t, 0] \; \rightarrow\; [x_t, y_{t}]$, where $y_{t}$ is
formed by sampling $\pi(y_{t} \mid x_t)$;
\item $[x_t, y_{t}] \rightarrow [0, y_{t}]$;
\item If $y_{t} \in Y_{halt}$, the process ends;
\item $[0, y_{t}] \rightarrow [y_{t}, y_{t}]$;
\item $[y_{t}, y_{t}] \; \rightarrow\; [y_{t}, 0]$;
\item Return to (1) with $t$ replaced by $t+1$;
\end{enumerate}

If this process ends (at stage (3)) with
$t = \tau$, then the associated value $y_\tau$ is used to specify an action by the organism back on its environment.
At this point, to complete a thermodynamic cycle, both $x$ and $y$ are reinitialized to zero, in preparation for
a new input. 

Here, for simplicity, I do not consider the thermodynamics of the physical
system that sets the initial value of $b_0$ by ``sensing
the environment''; nor do I consider the thermodynamics of the physical
system that copies that value to $x_0$ (see~\cite{parrondo2015thermodynamics} and the references therein for
some discussion of the thermodynamics of copying). In~addition I do not analyze
the thermodynamics of the process in which the organism uses $y_\tau$ to ``take
an action back to its environment'' and thereby reinitializes $y$. I only calculate the minimal work required to
implement the phenotype of the organism, which here is taken to mean the iterated ping-pong sequence between $X$ and $Y$. 

Moreover, I do not make any assumption for what happens to $b_0$ after it is used to
set $x_1$; it may stay the same, may slowly decay in some way, \textit{etc}. Accordingly,
none of the thermodynamic processes considered below are allowed to exploit
(some assumption for) the value of $b$ when they take place to reduce the amount
of work they require. As a result, from now on, I ignore the input space and its~partition.





Physically, a ping-pong sequence is implemented by some continuous-time 
stochastic processes over $W \times U$. Any such process induces an associated discrete-time 
stochastic process over $W \times U$. That discrete-time process comprises
a joint distribution $Pr$ defined over a (possibly
infinite) sequence of values $(w_0, u_0), \ldots (w_{t}, u_t), (w_{t+1}, u_{t+1}), \ldots$ That distribution in turn induces a 
joint distribution over associated pairs of partition element labels, $(w_0, u_0), \ldots (x_t, y_t), (x_{t+1}, y_{t+1}), \ldots$

%


For calculational simplicity, I assume that $\forall y \in Y$, at the end of each stage in
a ping-pong sequence that starts at any time $t \in \N$, $Pr(u \mid y)$
is the same distribution, which I write as $q^y_{out}(u)$. I make the ana<us
assumption for $Pr(w \mid x)$ to define $q^x_{proc}(w)$
(in addition to simplifying the analysis, this helps ensure that we are considering cyclic
processes, a crucial issue whenever analyzing issues like the minimal amount of
work needed to implement a desired map). Note that $q^y_{out}(u) = 0$ if $\Y(u) \ne y$.
To simplify the analysis further, I also assume that all ``internal entropies'' of the processor macrostates are the same,
\textit{i.e.}, $S(q^y_{out}(U))$ is independent of $y$, and similarly for the internal entropies of the output macrostates.


Also for calculational simplicity, I assume that at the end of each stage in
a ping-pong sequence that starts at any time $t \in \N$, there is no interaction Hamiltonian coupling any of the three
subsystems (though obviously, there must be such coupling at non-integer times).
I also assume that at all such moments, the Hamiltonian over $U$ is the same function, which 
I write as $H_{out}$. Therefore, for all such moments, the expected value of the Hamiltonian over $U$ if the 
system is in state $y_t$ at that time is:
\ba
\E(H_{out} \mid y) &=& \sum_u q^{y}_{out}(u) H_{out}(u)
\label{eq:expected_H}
\ea

Similarly, $H_{in}$ and $H_{proc}$ define the Hamiltonians at all such moments, over the input and processor
subsystems, respectively.

I will refer to any quadruple $(W, \X, U, \Y)$ and three associated Hamiltonians as an \textit{organism}.

For future use, note
that for any iteration $t \in \N$, initial distribution $\PP'(x_1)$, conditional distribution $\pi(y \mid x)$ and halting subset $Y_{halt} \subseteq Y$,
\ba
\PP'(y_t \in Y_{halt}) &=& \sum_{y_t} \PP'(y_t) I(y_t \in Y_{halt}) \nonumber \\
 &=& \sum_{x_t, y_t} \PP'(x_t) \pi(y \mid x)|_{x=x_t,y=y_t} I(y_t \in Y_{halt})
\ea
\ba
\PP'(y_t \mid y_t \in Y_{halt}) &=& \frac{\sum_{x_t}\PP'(x_t) \pi(y \mid x)|_{x=x_t,y=y_t} I(y_t \in Y_{halt})}
{\sum_{x_t,y_t} \PP'(x_t)\pi(y \mid x)|_{x=x_t,y=y_t} I(y_t \in Y_{halt})} \nonumber \\
\ea
and similarly:
\ba
 \PP'(x_{t+1} \mid y_t \not\in Y_{halt}) &=& 
 \frac{\sum_{x_t}\PP'(x_t) \pi(y \mid x)|_{x=x_t,y=x_{t+1}} I(x_{t+1} \not\in Y_{halt})} 
{\sum_{x_t,y_t} \PP'(x_t) \pi(y \mid x)|_{x=x_t,y=x_{t+1}} I(x_{t+1} \not\in Y_{halt})}  \nonumber \\
\ea

Furthermore,
\ba
S(\PP_t(X)) &=& -\sum_x \PP_t(x) \ln[\PP_t(x)] \\
S(\PP_{t+1}(X)) &=& -\sum_{x,y} \PP_t(x) \pi(y \mid x) \ln \bigg[\sum_{x'} \PP_t(x') \pi(y \mid x') \bigg]  \nonumber \\
\ea

I end this subsection with some notational comments. 
I will sometimes abuse notation and put time indices on distributions
rather than variables, e.g., writing $Pr_t(y)$ rather than $Pr(y_t = y)$.
In addition, sometimes, I abuse notation with temporal subscripts. In particular, when the initial distribution over $X$ is $\PP_1(x)$,
I sometimes use expressions like:
\ba
\PP_t(w) &\equiv& \sum_{x} \PP_t(x) q^{x}_{in}(w) \\
\PP_t(u) &\equiv& \sum_{y} \PP_t(y) q^{y}_{out}(u) \\
\PP_{t}(y) &\equiv& \sum_{x_t} \PP_t(x_t) \pi(y_{t} \mid x_t) \\
\PP_{t+1}(x \mid y_t) &\equiv& \delta(x, y_t) 
\ea

However, I will always be careful when writing joint distributions over
variables from different moments of time, e.g., writing:
\ba
\PP(y_{t+1}, x_t) &\equiv& \PP(y_{t+1} \mid x_t) \PP(x_t) \nonumber \\
&=& \pi(y_{t+1} \mid x_t) \PP_t(x_t)
\ea




\subsection{The Thermodynamics of Mapping an Input Space to an Output Space}
\label{sec:thermo_input_to_output}

Our goal is to construct a physical process $\Lambda$ over an organism's quadruple $(W, \X, U, \Y)$ that implements
an iteration of a given ping-pong sequence above for any particular $t$.
In addition, we want $\Lambda$ to be thermodynamically optimal with the stipulated starting and ending joint
Hamiltonians for all iterations of the ping-pong sequence when it is run on an initial joint distribution: 
\ba
\PP_1(x, y) = \PP_1(x) \delta(y, 0)
\label{eq:start}
\ea

%

In Appendix B, I present four separate GQ processes that implement stages (1), (2), (4) and (5) in a ping-pong sequence
(and so implement the entire sequence). The GQ processes for stages (1), (4) and (5)
are guaranteed to be thermodynamically reversible, for all $t$. However, each time-$t$ GQ process for stage (2) is
parameterized by a distribution $\G_t(x_t)$. Intuitively, that distribution is a guess, made by the ``designer''
of the (time-$t$) \mbox{stage (2)} GQ process, for the marginal distribution over the values $x_t$ at the beginning of 
the associated stage (1) GQ process. That stage (2) GQ process will also be
thermodynamically reversible, if the distribution over $x_t$ at the beginning of 
the stage (1) GQ process is in fact $\G_t(x_t)$. Therefore, for that input distribution, the sequence of GQ processes
is thermodynamically optimal, as desired. However,
as discussed below, in general, work will be dissipated if the stage (2) GQ process 
is applied when the distribution over $x_t$ at the beginning of stage (1) differs from $\G(x_t)$.

I call such a sequence of five processes implementing an iteration of a ping-pong sequence an \textit{organism process}.
It is important to emphasize that I do \emph{not} assume that any particular real biological system runs an organism
process. An organism process provides a counterfactual model of how to implement
a particular dynamics over $X \times Y$, a model that allows us to calculate
the minimal work used by any actual biological system that implements that dynamics.

Suppose that an organism process always halts for any $x_1$, such that $\PP_1(x_1) \ne 0$.
Let $\tau^*$ be the last iteration
at which such an organism process may halt, for any of the inputs $x_1$, such that $\PP(x_1) \ne 0$
(note that if $X$ is countably infinite, $\tau^*$ might be countable infinity). Suppose further
that no new input is received before $\tau^*$ if the process halts at some $\tau < \tau^*$ and
that all microstates are constant from such a $\tau$ up to $\tau^*$ (so, no new
work is done during such an interval). In light of the iterative
nature of organism processes, this last assumption is equivalent to
assuming that $\pi(y_t \mid x_t) = \delta_{y_t, x_t}$ if $x_t \in Y_{halt}$. 

I say that the organism process is \textit{recursive} when all of these conditions are met, since that is the adjective used in
the theory of Turing machines. For a recursive organism process, the ending distribution over $y$ is:
\ba
\PP(y_{\tau^*}) &=& \sum_{x_1, \ldots, x_{\tau^*}} \pi(y_{\tau^*} \mid x_{\tau^*}) \PP_1(x_1) \prod_{t=1}^{\tau^*} \pi(x_{t} \mid x_{t-1})
\ea
and:
\ba
\PP(y_{\tau^*} \mid x_1) &=& \sum_{x_2, \ldots, x_{\tau^*}} \pi(y_{\tau^*} \mid x_{\tau^*}) \prod_{t=1}^{\tau^*} \pi(x_{t} \mid x_{t-1})
\label{eq:conditional}
\ea

\begin{proposition}
\label{prop:1}
Fix any recursive organism process, iteration $t \in \N$,
initial distributions $\PP_1(x), \PP'_1(x)$, conditional distribution $\pi(y \mid x)$ and halting subset $Y_{halt} \subseteq Y$.
\begin{enumerate}
\item With probability $\PP'(y_t \in Y_{halt})$,
the ping-pong sequence at iteration $t$ of the associated organism process maps the distribution:
\ba
\PP'(x_t)\delta(y_{t-1},0) &\rightarrow& \delta(x_{t},0) \PP'(y_t \mid y_t \in Y_{halt}) \nonumber
\ea
and then halts, and with probability $1 - \PP'(y_t \in Y_{halt})$, it instead maps:
\ba
 \PP'(x_t)\delta(y_{t-1},0) &\rightarrow& \PP(x_{t+1} \mid y_t \not\in Y_{halt}) \delta(y_t, 0)
%
 \nonumber
\ea
and continues.

\item 
If $\G_t = \PP_t$ for all $t \le \tau^*$,
the total work the organism expends to map the initial distribution $\PP_1(x)$
to the ending distribution $\PP_{\tau^*}(y)$ is:
%
%
%
%
%
%
\ba
\WW^\pi_{\PP_1} &\equiv&
\sum_{y} \PP_{\tau^*}(y) \E(H_{out} \mid y) - \E(H_{out} \mid y')|_{y' = 0} \nonumber \\
&&  - \;  \sum_{x} \ \PP_{1}(x) \E(H_{in} \mid x) \;+\; \E(H_{in} \mid x')|_{x' = 0} \nonumber \\
		&& +\; kT \big(S(\PP_1(X)) - S(\PP_{\tau^*}(Y))\big) \nonumber
\label{eq:fillin}
\ea

\item There is no physical process that  
both performs the same map as the organism process and that requires
less work than the organism process does when applied to $\PP(x_t)\delta(y_t,0)$.
\end{enumerate}
\end{proposition}

\begin{proof}
Repeated application of Proposition~\ref{prop:22} gives the first result.

Next, combine~Equation (\ref{eq:pi_partial}) in Appendix B,~Equation (\ref{eq:gq_work_required}) and
our assumptions made just before Equation~(\ref{eq:expected_H}) to 
calculate the work needed to implement the GQ process of the first stage of an organism process at iteration $t$:
\ba
&&	\bigg[\sum_{x,y,u} \bigg( \PP_t(x) \pi(y \mid x)q^y_{out}(u) - q^0_{out}(u)\bigg) H_{out}(u) \bigg]
		 - kT\bigg[S(\PP_{t}(Y)) - S(\PP_{t-1}(Y))\bigg] \nonumber \\
&& = \; 	\sum_y \PP_t(y) \E(H_{out} \mid y) \;-\; \E(H_{out} \mid y')|_{y'=0} \;-\; kTS(\PP_{t}(Y)) \nonumber
\ea
Analogous equations give the work for the remaining three GQ processes.
Then, apply these equations repeatedly, starting with 
the distribution given in~Equation (\ref{eq:start}) (note that all terms for iterations of the ping-pong sequence
with $t \in \{2, 3, \ldots, \tau^* - 1\}$ cancel out). This gives the second result.

Finally, the third result is immediate from the assumption that $\G_t = \PP_t$ for all $t$,
which guarantees that each iteration of the organism process is thermodynamically
reversible.
\end{proof}

The first result in Proposition~\ref{prop:1} means that no matter what the initial distribution over $X$ is, the 
organism process updates that distribution according to $\pi$, halting whenever it produces a value in $Y_{halt}$. 
This is true even if the output of $\pi$ depends on its input (as discussed in
the Introduction, this property is violated for many of the physical processes considered in the literature).

The first terms in the definition of $\WW^\pi_{\PP_1}$, given by a sum of expected values of the Hamiltonian,
can be interpreted as the ``labor'' done by the organism when processing $x_1$ into $y_{\tau^*}$,
e.g., by making and breaking chemical bonds. It quantifies the minimal amount of external
free energy that must be used to implement the amount of labor that is (implicitly) specified
by $\pi$. The remaining terms, a difference
of entropies, represent the free energy required by the ``computation'' done by the organism when it 
undergoes $\pi$, independent of the labor done by the organism.

\subsection{Input Distributions and Dissipated Work}
\label{sec:dissipated_work}


Suppose that at the beginning of some iteration $t$ of an organism process, the distribution over $x_t$ is some $\PP(x_t)$ that differs from
$\G_t(x_t)$, the prior distribution ``built into'' the (quenching Hamiltonians defining the)
organism process. Then, as elaborated at the beginning of Section~\ref{sec:thermo_input_to_output}, in general,
this iteration of the organism process will result in dissipated work.

As an example, such dissipation will occur if the organism process is used in an environment that generates inputs according
to a distribution $\PP_1$ that differs from $\G_0$, the distribution ``built into'' the
organism process. In the context of biology, if a biological system gets optimized by 
natural selection for one environment, but is then used in another one, it will necessarily
operate (thermodynamically sub-optimally) in that second environment.

Note though that one could imagine designing an organism to operate optimally for a distribution
over environments, since that is equivalent to a single average distribution
over inputs. More precisely, a distribution $Pr(\PP_1)$ over environments is equivalent
to a single environment generating inputs according to:
\ba
Pr(x_1) &=& \sum_{\PP_1} Pr(\PP_1) \PP_1(x_1)
\ea
We can evaluate the thermodynamic cost $\Omega^\pi_{Pr}$ for this organism that behaves optimally
for an uncertain environment.

As a comparison point, we can also evaluate the work used in an impossible
scenario where $\PP_1$ varies stochastically but the organism magically ``knows'' what each $\PP_1$ is
before it receives an input sampled from that $\PP_1$, and then changes its distributions $\G_t$
accordingly. The average thermodynamic cost in this impossible scenario would be
\ba
\sum_{\PP_1} Pr(\PP_1) \Omega^\pi_{\PP_1}
\ea
In general 
\ba
\Omega^\pi_{Pr} &\ge& \sum_{\PP_1} Pr(\PP_1) \Omega^\pi_{\PP_1}
\ea
with equality only if $Pr(.)$ is a delta function about one particular $\PP_1$. So in general,
even if an organism choose its (fixed) $\G_0$ to be optimal for an uncertain environment, it cannot
do as well as it would if it could magically change $\G_0$ appropriately before each new environment it encounters.

As a second example, in general, as one iterates an organism process, the initial distribution $\PP_1(x)$ 
is changed into a sequence of new distributions $\{\PP_1(x), \PP_2(x), \ldots \}$. In general,
many of these distributions will differ, \textit{i.e.}, for many $t'$, $\PP_{t'+1} \ne \PP_{t'}$.
Accordingly, if one is using some particular physical device to implement the organism process, unless that
device has a clock that it can use to update $\G_t$ from one iteration to the next (to~match the changes in $\PP_t$),
the distribution $\G_t$ built into the device will differ from $\PP_t$ at some times $t$. Therefore, without such a clock, work will
be dissipated.

Bearing these caveats in mind, unless explicitly stated otherwise, in the sequel, I assume that 
the time-$t$ stage (2) GQ process of an organism makes the correct guess
for the input distribution at the start of the time-$t$ ping-pong sequence, \textit{i.e.},
that its parameter $\G_t$ is always the same as 
the distribution over $x$ at the beginning of the time-$t$ stage (1) process.
In this case, the minimal free energy required by the organism is $\WW^\pi_{\PP_1}$,
and no work is dissipated.

It is important to realize that in general, if one were to run a Q process over $X$ in the second stage of an
organism process, rather than
a GQ process over $X$ guided by $Y$, there would be nonzero dissipated work. The~reason is that if we ran such a Q process,
we would ignore the information in $y_{t+1}$
concerning the variable we want to send to zero, $x_t$. 
In contrast, when we use a GQ process over $X$ guided by $Y$, no information is ignored, and
we maintain thermodynamic reversibility. The extra work of the Q process beyond that of the GQ process is:
\ba
kTS(X_t) - kTS(X_t \mid Y_{t+1}) &=& kTI(X_t ; Y_{t+1})
\label{eq:mutual_info}
\ea

In other words, using the Q process would cause us to dissipate work $kT I(X_t ; Y_{t+1})$. This amount of
dissipated work equals zero if the output of $\pi$ is independent of its input, as in bit erasure. It also
equals zero if $P(x_t)$ is a delta function.
However, for other $\pi$ and $P(x_t)$, that dissipated work will be nonzero.
In such situations, stage 2 would be thermodynamically \emph{ir}reversible if we used a
Q process over $X_t$ to set $x$ to zero. 

As a final comment, it is important to emphasize that no claim is being made that the only
way to implement an organism process is with Q processes and/or GQ processes.
However, the need to use the organism process in an appropriate environment,
and for it to have a clock, should be generic, if we wish to avoid dissipated work.

\subsection{Optimal Organisms}

From now on, for simplicity, I restrict attention to recursive organism processes.

Recall that adding noise to $\pi$ may reduce the amount of work required to implement it.
Formally, Proposition~\ref{prop:1} tells us that everything else being equal, the larger $S(\PP_{\tau^*}(Y))$ is,
the less work is required to implement the associated $\pi$
(indeed, the thermodynamically-optimal implementation of a one-to-many map $\pi$ actually draws in free energy
from the heat bath, rather than requiring free energy that ends up being dumped into that heat bath). This implies that an organism will
want to implement a $\pi$ that is as noisy as possible. 

In addition, not all maps $x_1 \rightarrow y_{\tau^*}$ are equally important to an organism's reproductive
fitness. It will be important to be very precise in what output is produced for some
inputs $x_1$, but for other inputs, precision is not so important. Indeed, for some
inputs, it may not matter at all what output the organism produces in response.

In light of this, natural selection would be expected to favor $\pi$'s that are as noisy
as possible, while still being precise for those inputs where reproductive fitness requires it.
To simplify the situation, there are two contributions to the reproductive fitness of
an organism that implements some particular $\pi$: the free energy (and other resources)
required by that implementation and the ``phenotypic fitness'' that would arise
by implementing $\pi$ even if there were no resources required to implement it.

Therefore, there will be a tradeoff between the resource cost of
being precise in $\pi$ with the phenotypic fitness benefit of being precise. In particular, there
will be a tradeoff between the thermodynamic cost of being precise in $\pi$ (given by
the minimal free energy that needs to be used to implement $\pi$) and the phenotypic
fitness of that $\pi$.
In this subsection, I use an extremely simplified and abstracted model of reproductive fitness
of an organism to determine what $\pi$ optimizes this~tradeoff.

To start, suppose we are given a real-valued \textit{phenotypic fitness} function $f(x_1, y_{\tau^*})$. This
quantifies the benefit to the organism of being precise in what output it
produces in response to its inputs. More precisely, $f(x_1, y_{\tau^*})$
quantifies the impact on the reproductive fitness of the organism that arises 
if it outputs $y_{\tau^*}$ in response to an input $x_1$ it received, minus
the effect on reproductive fitness of how the organism generated that response. 
That second part of the definition means that behavioral fitness does
not include energetic costs associated with mapping $x_1 \rightarrow y_{\tau^*}$. Therefore, it includes neither the
work required to compute a map taking $x_1 \rightarrow y_{\tau^*}$ nor the labor
involved in carrying out that map going into $f$ (note that in some toy models, $f(x_1, y_{\tau^*})$
would be an expectation value of an appropriate quantity, taken over states of the environment, and conditioned 
on $x_1$ and $y_{\tau^*}$). 
For an input distribution $\PP_1(x)$ and conditional distribution $\pi$, expected phenotypic fitness is:
\ba
\E_{\PP_1, \pi}(f) &=& \sum_{x_1,y_{\tau^*}} \PP_1(x_1) \PP(y_{\tau^*} \mid x_1) f(x_1, y_{\tau^*})
\ea
where $\PP(y_{\tau^*} \mid x_1)$ is given by~Equation (\ref{eq:conditional}).

The expected phenotypic fitness of an organism if it implements $\pi$ on the initial
distribution $\PP_1$ is only one contribution to the overall reproductive fitness of the
organism. In addition, there is a reproductive fitness cost to the organism 
that depends on the specific physical process it uses to implement $\pi$ on $\PP_1$. In particular,
there is such a cost arising from the physical resources that the process requires. 

There are several contributions to this cost. In particular,
different physical processes for implementing $\pi$ will require different sets of chemicals from the
environment, will result in different chemical waste products, \textit{etc}. Here, I ignore
such ``material'' costs of the particular physical process the organism uses to implement \mbox{$\pi$ on $\PP_1$. }

However, in addition to 
these material costs of the process, there is also a cost arising from the thermodynamic work required to run that process.
If we can use a thermodynamically-reversibly process, then by~Equation (\ref{eq:fillin}), for fixed $\PP_1$ and $\pi$, the 
minimal possible such required work is $\WW^\pi_{\PP_1}$. Of course, in many biological scenarios,
it is not possible to use a thermodynamically-reversible organism process to implement~$\pi$. As discussed in
Section~\ref{sec:dissipated_work}, this is the case if the organism process is ``designed'' for
an environment that generates inputs $x$ according to $\G_1(x)$ while the actual environment in which the process is used
generates inputs according to some $\PP_1 \ne \G_1$. However, there are other reasons
why there might have to be non-zero dissipated work. In particular, there is non-zero dissipated
work if $\pi$ must be completed quickly, and so, it cannot be implemented using a quasi-static process 
(it does not do an impala any good to be able to compute the optimal direction in which to flee a tiger chasing it,
if it takes the impala an infinite amount of time to complete that computation). Additionally, of course, it may
be that a minimal amount of work must be dissipated simply because of the limited kinds of 
biochemical systems available to a real organism.

I make several different simplifying assumptions:
\begin{enumerate}
\item In some biological scenarios, the amount of
such dissipated work that cannot be avoided in implementing $\pi$, $\hat{W}^\pi_{\PP_1}$, will be comparable
to (or even dominate) the minimal amount of reversible work needed to implement $\pi$, $\WW^\pi_{\PP_1}$. 
However, for simplicity, in the sequel, I concentrate solely on the dependence on $\pi$ of the reproductive fitness of a process that
implements $\pi$ that arises due to its effect on ${W}^\pi_{\PP_1}$. Equivalently,
I assume that I can approximate differences
$\hat{W}^\pi_{\PP_1} - \hat{W}^{\pi'}_{\PP_1}$ as equal to $\hat{W}^\pi_{\PP_1} - \hat{W}^{\pi'}_{\PP_1}$
up to an overall proportionality~constant.

\item Real organisms have internal energy stores that allow them to use free energy extracted from the environment at a time $t' < 1$ to drive
a process at time $t=1$, thereby ``smoothing out'' their free energy needs. For simplicity, I ignore such energy stores.
Under this simplification, the organism needs to extract at least $\WW^\pi_{\PP_1}$ of free energy
from its environment to implement a single iteration of $\pi$ on $\PP_1$. That minimal amount of needed free energy
is another contribution to the ``reproductive fitness cost to the organism of physically implementing $\pi$ 
starting from the input distribution $\PP_1$''. 
%

\item As another simplifying assumption, I
suppose that the (expected) reproductive fitness of an organism that implements the map $\pi$
starting from $\PP_1$ is just:
\ba
\mathscr{F}(\PP_1, \pi, f) &\equiv& \alpha \E_{\PP_1, \pi}(f) - \WW^\pi_{\PP_1}
\ea
Therefore, $\alpha$ is the benefit to the organism's reproductive fitness of increasing $f$ by one, measured in
units of energy.
%
This ignores all effects on the distribution $\PP_1$ 
that would arise by having different $\pi$ implemented at times earlier than $t=1$. It also ignores the possible
impact on reproductive fitness of the organism's implementing particular sequences of
multiple $y$'s (future work involves weakening all of these assumptions, with particular attention to this last one).
Under this assumption, varying $\pi$ has no effect on $S(X_1)$, the initial entropy over processor states.
Similarly, it has no effect on the expected value of the Hamiltonian then.
\end{enumerate}

%
%
%
%
%

Combining these assumptions with Proposition~\ref{prop:1}, we see that after removing all terms in $\WW^\pi_{\PP_1}$ that 
do not depend on $\pi$, we are left with
$\sum_{y} \PP_{\tau^*}(y) \E(H_{out} \mid y) - \;kTS(\PP_{\tau^*}(Y))$. This gives the following result:
\begin{corollary}
\label{prop:2}
Given the assumptions discussed above, up to an additive constant that does not depend on $\pi$:
%
\ba
&& \mathscr{F}(\PP_1, \pi, f) \nonumber \\
&& \qquad = \sum_{x_1, y_{\tau^*}} \PP(x_1) \PP(y_{\tau^*} \mid x_1) \bigg\{ \alpha f(x_1, y_{\tau^*}) - H_{out}(y_{\tau^*}) \nonumber \\
&& \qquad \qquad \qquad \qquad \qquad - kT\ln \bigg[\sum_{x'_1} \PP_1(x'_1) \PP(y_{\tau^*} \mid x'_1) \bigg] \bigg\} \nonumber
\label{eq:actual_minimal_work}
\ea
\label{coroll:1}
\end{corollary}

\vspace{-14pt}The first term in Corollary~\ref{coroll:1} reflects the impact of $\pi$ on the phenotypic fitness of the organism. The second term 
reflects the impact of $\pi$ on the amount of labor the organism does. Finally, the last term 
reflects the impact of $\pi$ on the amount of computation the organism does; the greater the entropy of $y_{\tau^*}$, the less total computation is done. In different biological
scenarios, the relative sizes of these three terms may change radically. In some senses, Corollary~\ref{coroll:1}
can be viewed as an elaboration of~\cite{kussell2005phenotypic}, where the ``cost of sensing'' constant
in that paper is decomposed into labor and computation costs.

From now on, for simplicity, I assume that $Y_{halt} = Y$. So no matter
what the input is, the organism process runs $\pi$ exactly once to produce the output.
Returning to our actual optimization problem, 
by Lagrange multipliers, if the ${\pi}$ that maximizes the expression in Corollary~\ref{prop:2}
lies in the interior of the feasible set, then it is the solution
to a set of coupled nonlinear equations, one equation for each pair $(x_1, y_{1})$:
\ba
&& \PP(x_1) \bigg\{H_{out}(y_{1}) - \alpha f(x_1, y_{1}) \nonumber \\
&& \qquad \qquad + kT \bigg(\ln \bigg[\sum_{x'_1} \PP(x'_1) {\pi}(y_{1} \mid x'_1) \bigg] + 1 \bigg)\bigg\} \;=\; \lambda_{x_1}
\ea	 
where the $\lambda_{x_1}$ are the Lagrange multipliers ensuring that $\sum_{y_{1}} {\pi}(y_{1} \mid x_1) = 1$ for all $x_1 \in X$.
Unfortunately in general the solution may not lie in the interior, so that we have a
non-trivial optimization problem.

However, suppose we replace the quantity:
\ba
 -\sum_{x_1, y_{1}} \PP_1(x_1) {\pi}(y_{1} \mid x_1)
	\ln \bigg[\sum_{x'_1} \PP_1(x'_1) {\pi}(y_{1} \mid x'_1) \bigg] &=& S(Y_{1}) \nonumber \\
\ea
in Corollary~\ref{prop:2} with $S(Y_{1} \mid X_1)$. Since $S(Y_{1} \mid X_1) \le S(Y_{1})$~\cite{coth91,mack03},
this modification gives us a lower bound on expected reproductive fitness:
\ba
\hat{\mathscr{F}}(\PP_1, {\pi}, f) &\equiv& \sum_{x_1, y_{1}} \PP_1(x_1) {\pi}(y_{1} \mid x_1)\bigg\{ \alpha f(x_1, y_{1}) \nonumber \\
&& \qquad - H_{out}(y_{1}) - kT \ln \bigg[ {\pi}(y_{1} \mid x_1) \bigg] \bigg\} \nonumber \\
		&\le& {\mathscr{F}}(\PP_1, {\pi}, f) 
\ea
%
%

The ${\pi}$ that maximizes $\hat{\mathscr{F}}(\PP_1, {\pi}, f)$ is just a set of Boltzmann distributions:
\ba
{\pi}(y_{1} \mid x_1) &\propto& \exp\bigg( \frac{\alpha f(x_1, y_{1}) - H_{out}(y_{1})}{kT}\bigg)
\ea

For each $x_1$, this approximately optimal conditional distribution puts more weight on
$y_{1}$ if the associated phenotypic fitness is high, while putting less weight on $y_{1}$ if the associated energy is large.
In addition, we can use this distribution to construct a lower bound on the maximal value of the expected reproductive fitness:
\begin{corollary}
Given the assumptions discussed above,
\ba
\max_{{\pi}} \mathscr{F}(\PP_1, {\pi}, f) &\ge& -kT \sum_{x_1} \PP(x_1) \ln \bigg[ \sum_{y_{1}} 
		\exp\bigg( \frac{\alpha f(x_1, y_{1}) - H_{out}(y_{1})}{kT}\bigg)\bigg] \nonumber
\ea
\end{corollary}
\begin{proof}
Write:
\ba
\hat{\mathscr{F}}(\PP_1, {\pi}, f) &=& \sum_{x_1, y_{1}} \PP_1(x_1) \bigg( {\pi}(y_{1} \mid x_1)\bigg\{ \alpha f(x_1, y_{1}) - 
	H_{out}(y_{1}) \nonumber \\
	&& \qquad \qquad \qquad \qquad- kT \ln \bigg[ {\pi}(y_{1} \mid x_1) \bigg] \bigg\} \bigg) \nonumber \\
	 &\equiv& \sum_{x_1, y_{1}} \PP_1(x_1) \hat{\mathscr{F}}(x_1, {\pi}, f)
\label{eq:intermediate}	
\ea
Each term $\hat{\mathscr{F}}(x_1, {\pi}, f)$ in the summand depends on the $Y$-space distribution $\pi(. \mid x_1)$, but
no other terms in $\pi$. Therefore, we can evaluate each such term $\hat{\mathscr{F}}(x_1, {\pi}, f)$ separately for its maximizing (Boltzmann)
distribution $\pi(. \mid x_1)$. In the usual way, this is given by the log of the associated partition function (normalization constant) $z(x_1)$,
since for any $x_1$ and associated Boltzmann $\pi(. \mid x_1)$,
\ba
S(Y_1 \mid x_1) &=& -\sum_{y_1} {\pi}(y_{1} \mid x_1) \ln[ {\pi}(y_{1} \mid x_1)] \nonumber \\
	&=& -\sum_{y_1} \frac{\exp \big(\beta [ \alpha f(x_1, y_{1}) - H_{out}(y_{1})] \big)}{z(x_1)} \nonumber \\
&& \qquad \qquad  \qquad \qquad	 \ln \bigg[ \frac{\exp \big(\beta [ \alpha f(x_1, y_{1}) - H_{out}(y_{1})] \big)}{z(x_1)} \bigg] \nonumber \\
	 &=& -\sum_{y_1} {\pi}(y_{1} \mid x_1) \big(\beta [ \alpha f(x_1, y_{1}) - H_{out}(y_{1})] \big) \;-\; \ln[z(x_1)] \nonumber \\
\ea
where $\beta \equiv 1 / kT$, as usual.
Comparing to~Equation (\ref{eq:intermediate}) establishes that:
\ba
\hat{\mathscr{F}}(x_1, {\pi}, f) &=& -kT \ln[z(x_1)]
\ea
and then gives the claimed result.
\end{proof}

%


As an aside, suppose we had $X = Y$, $f(x, x) = 0$ for all $x$ and that $f$ were non-negative. Then if in addition the amount of expected work were given by the mutual information between $X_1$ and $Y_1$ rather than the difference in their entropies, our optimization problem would reduce to finding a point
on the rate-distortion curve of conventional information theory, with $f$ being the distortion function~\cite{coth91}. (See also~\cite{taylor2007information} for a slight variant of rate-distortion theory, appropriate when $Y$ differs from $X$, and so the requirement that $f(x, x) = 0$ is dropped.)
However as shown above the expected work to implement $\pi$ does not depend on the precise coupling between $x_1$ and $y_1$ under $\pi$, but only the associated marginal distributions. So rate-distortion theory does not directly apply.

On the other hand, some of the same kinds of analysis used in rate-distortion theory can also be applied here. In particular, for any particular component $\pi(y_1 \mid x_1)$ where $\PP_1(x_1) \ne 0$, since $\tau^* = 1$,
\ba
\frac{\partial^2}{\partial \; \pi(y_1 \mid x_1)^2} {\mathscr{F}}(x_1, {\pi}, f) &=& \frac{\PP_1(x_1)}{\PP_1(y_1)} \nonumber \\
       &>& 0
 \ea
(where $\PP(y_1) = \sum_{x'_1} \PP(x'_1) \pi(y_1 \mid x_1)$, as usual). So
${\mathscr{F}}(x_1, {\pi}, f)$ is concave in every component of $\pi$. This means that the optimizing channel $\pi$ 
may lie on the edge of the feasible region of conditional distributions. 
Note though that even if the solution
is on the edge of the feasible region, in general for different $x_1$ that optimal
$\pi(y_1 \mid x_1)$ will put all its probability mass on different edges of the unit simplex over $Y$. So when
those edges are averaged under $\PP_1(x_1)$, the result is a marginal distribution $\PP(y_1)$ that lies
in the interior of the unit simplex over $Y$.

As a cautionary note, often in the real world, there is an inviolable upper bound on the rate at which
a system can ``harvest'' free energy from its environment, \textit{i.e.}, on how much free energy it can
harvest per iteration of ${\pi}$ (for example, a plant with a given surface
area cannot harvest free energy at a faster rate than sunlight falls upon its surface). In that case, we are not interested in optimizing
a quantity like $\mathscr{F}(\PP_1, {\pi}, f)$, which
is a weighted average of minimal free energy and expected phenotypic fitness per iteration of ${\pi}$. Instead,
we have a constrained optimization problem with an inequality constraint: find the ${\pi}$
that maximizes some quantity \mbox{(e.g., expected phenotypic fitness),} subject to an inequality constraint on the 
free energy required to implement that ${\pi}$. Calculating solutions to these kinds of constrained optimization problem
is the subject of future work.

\section{General Implications for Biology} 

Any work expended on an organism must first be acquired as free energy from the organism's environment. However, in
many situations, there is a limit on the flux of free energy through an organism's immediate environment.
Combined with the analysis above, such limits provide upper bounds on the ``rate of (potentially
noisy) computation'' that can be achieved by a biological organism in that environment, once all 
energetic costs for the organism's labor (\textit{i.e.}, its moving, making/breaking chemical bonds, \textit{etc}.) are accounted for.

As an example, human brains do little labor. Therefore, these results bound the rate of computation of a human brain. 
Given the fitness cost of such computation (the brain uses $\sim$20\% of
the calories used by the human body), this bound contributes to the natural selective pressures on humans
(in the limit that operational inefficiencies of the brain have already been minimized).
In other words, these bounds suggest that natural selection imposes a tradeoff between the fitness quality of a brain's decisions
and how much computation is required to make those decisions. In this regard, it is 
interesting to note that the brain is famously noisy, and as discussed above, noise in
computation may reduce the total thermodynamic work required 
(see~\cite{bullmore2012economy,sandberg2016energetics,laughlin2001energy} for more
about the energetic costs of the human brain and its relation to Landauer's bound).

As a second example, the rate of solar free energy incident upon the Earth 
provides an upper bound on the rate of computation that can be achieved by the biosphere
(this bound holds for any choice for the partition of the biosphere's fine-grained space
into macrostates, such that the dynamics over those macrostates executes $\pi$).
In particular, it provides an upper bound on the rate of computation that can be achieved by human civilization, if we remain
on the surface of the Earth and only use sunlight to power our computation. 

Despite the use of the term ``organism'', the analysis above is not limited to
biological individuals. For~example, one could take the input to be a current generation
population of individuals, together with attributes of the environment shared
by those individuals. We could also take the output to be the next
generation of that population, after selective winnowing based on
the attributes of the environment (e.g., via replicator dynamics). In this 
example, the bounds above do not refer to the ``computation'' performed
by an individual, but rather by an entire population subject to natural
selection. Therefore, those bounds give the minimal free energy required 
to run natural selection.

As a final example, one can use these results to analyze how the thermodynamic behavior of the
biosphere changes with time. In particular, if one iterates $\pi$ from one $t$ to the
next, then the associated initial distributions $\PP_t$ change. Accordingly, the minimal
amount of free energy required to implement $\pi$ changes. In theory, this allows
us to calculate whether the rate of free energy required by the information processing
of the terrestrial biosphere increases with time. Prosaically, has the rate of
computation of the biosphere increased over evolutionary timescales? If it has done so for most of the time
that the biosphere has existed, then one could plausibly view the fraction of free energy flux from
the Sun that the biosphere uses as a measure of the ``complexity'' of
the biosphere, a measure that has been increasing throughout the lifetime of the
biosphere. 

Note as well that there is a fixed current value of the total free energy flux incident on the biosphere
(from both sunlight and, to a much smaller degree, geologic processes). By the
results presented above, this rate of free energy flux gives an upper bound 
on the rate of computation that humanity as a whole can ever achieve, if it monopolizes
all resources of Earth, but restricts itself to the surface of Earth.
%

\section{Discussion}

The noisier the input-output map $\pi$ of a biological organism, the less free energy
the organism needs to acquire from its environment to implement that map. Indeed,
by using a sufficiently noisy $\pi$, an organism can \emph{increase} its stored free energy.
Therefore, noise might not just be a hindrance that an organism needs to circumvent; an organism
may actually exploit noise, to ``recharge its battery''.

In addition, not all maps $x_t \rightarrow y_{t+1}$ are equally important to an organism's reproductive
fitness. In light of this, natural selection would be expected to favor $\pi$'s that are as noisy
as possible, while still being precise for those inputs where reproductive fitness requires it.

In this paper, I calculated what $\pi$ optimizes this tradeoff. This calculation provides insight into
what phenotypes natural selection might be expected to favor. Note though that in the real world,
there are many other thermodynamic factors that are important in addition to the cost of
processing sensor readings (inputs) into outputs (actions). For example, there are the costs of
acquiring the sensor information in the first place and of internal storage of such information,
for future use. Moreover, in the real world, sensor readings do not arrive in an {i.i.d.} 
%
 basis, as assumed in this paper. Indeed, in real biological systems, often, the current sensor
reading, reflecting the recent state of the environment, reflects previous actions by the
organism that affected that same environment (in other words, real biological
organisms often behave like feedback controllers). All of these effects would modify
the calculations done in this paper.

In addition, in the real world, there are strong limits on how much time a biological system can take to perform its computations, physical labor and rearranging of matter, due to environmental exigencies (simply put, if the biological system is not fast enough, it may be killed). These temporal constraints mean that biological systems cannot use fully reversible thermodynamics. Therefore, these temporal constraints increase the free energy required for the biological system to perform computation, labor and/or rearrangement of matter.

Future work involves extending the analysis of this paper to account for such thermodynamic effects. Combined
with other non-thermodynamic resource restrictions that real biological organisms face, such
future analysis should help us understand how closely the organisms that natural selection
has produced match the best ones possible.

\newpage
\acknowledgments{{\bf Acknowledgment:} I would like to thank Daniel Polani, Sankaran Ramakrishnan and especially Artemy Kolchinsky
for many helpful discussions 
and the Santa Fe Institute for helping to support this research. This paper was made possible through the support of Grant No. TWCF0079/AB47 from the Templeton World Charity Foundation and Grant No.
FQXi-RHl3-1349 from the {FQXi} 
%
 foundation. 
The opinions expressed in this paper are those of the author and do not necessarily 
reflect the view of Templeton World Charity Foundation.}


\appendix
\section*{Appendix A: Proof of Proposition~\ref{prop:22}}

We begin with the following lemma:
\begin{lemma}
A GQ process over $R$ guided by $V$ 
(for conditional distribution $\pi$ and initial distribution $\rho^t(r, s)$) will transform
any initial distribution:
\ba
p^t(r, s) &=& \sum_v p^t(v) \rho^t(s \mid v) p^t(r \mid v)
\label{eq:31}
\ea
into a distribution:
\ba
p^{t+1}(r, s) &=& \sum_v p^t(v) \rho^t(s \mid v) \pi(r \mid v)
\label{eq:32}
\ea
\label{lemma:1}
\end{lemma}

\begin{proof}
Fix some $v^*$ by sampling $p^t(v)$. Since in a GQ, microstates only change during the quasi-static relaxation,
after the first quench, $s$ and, therefore, $v$ still equal $v^*$. Due to the infinite potential barriers in $\mathcal{S}$, while
$s$ may change during that relaxation, $v$ will not, and so, $v^{t+1} = v^* = v_t$. Therefore: 
\ba
H^t_{quench; int}(r, s) &\equiv& -kT \ln[\pi(r \mid v_t)]
\label{eq:lemma}
\ea
Now, at the end of the relaxation step, $\rho(r, s)$ has settled to thermal equilibrium
within the region $R \times v_t \subset R \times V$. Therefore,
combining~Equation (\ref{eq:lemma}) with~Equations (\ref{eq:29}) and~(\ref{eq:28}), we see that the distribution at the end of the relaxation~is:
\ba
\rho^{t+1}(r, s) &\propto& \exp{\big(\frac{-H^{t+1}_{quench}(r, s) }{kT}\big)} \; \delta(V(s), v_t) \nonumber \\
  &=& \exp{\big( \ln[\pi(r \mid v_t)] + \ln[\rho^t(s)]\big)} \; \delta(V(s), v_t) \nonumber \\
  &=& \pi(r \mid v_t) \rho^t(s) \delta(V(s), v_t) \nonumber \\
  &\propto& \pi(r \mid v_t) \rho^t(s \mid v)
\ea
Normalizing,
\ba
\rho^{t+1}(r, s) &=& \pi(r \mid v_t) \rho^t(s \mid v)
\ea
Averaging over $v_t$ then gives $p^{t+1}(r, s)$:
\ba
p^{t+1}(r, s) &=& \sum_v p^t(v) \rho^t(s \mid v) \pi(r \mid v)
\ea
\end{proof}

Next, note that
$\rho^t(s \mid v) = 0$ if $s \not \in V(s)$. Therefore, if~Equation (\ref{eq:32}) holds and we sum $p^{t+1}(r, s)$ over all $s \in V^{-1}(v)$
for an arbitrary $v$, we get:
\ba
p^{t+1}(r, v) &=& p^t(v)\pi(r \mid v)
\ea
%
Furthermore, no matter what $\rho^t(s \mid v)$ is, $p^t(r, v) = p^t(v) p^t(r \mid v)$.
As a result, Lemma~\ref{lemma:1} implies that a GQ process over $R$ guided by $V$ 
(for conditional distribution $\pi$ and initial distribution $\rho^t(r, s)$) will transform
any initial distribution $p^t(v) p^t(r \mid v)$ into a distribution $p^t(v) \pi(r \mid v)$.
This is true whether or not $p^t(v)= \rho^t(v)$ or $p^t(r \mid v) = \rho^t(r \mid v)$. 
This establishes the claim of Proposition~\ref{prop:22} that the first ``crucial feature'' of GQ processes~holds.


\section*{Appendix B: The GQ Processes Iterating a Ping-Pong Sequence}

In this section, I present the separate GQ processes for implementing the stages of a ping-pong sequence.

First, recall our assumption from just below the definition of a ping-pong 
sequence that at the end of any of its stages,
$Pr(u \mid y)$ is always the same distribution $q^y_{out}(u)$ (and similarly for distributions
like $Pr(w \mid x)$). Accordingly, at the end of any stage of a ping-pong sequence that implements 
a GQ process over $U$ guided by $X$, we can uniquely recover the conditional distribution $Pr(u \mid x)$ from~$Pr(y \mid x)$:
\ba
{\overline{\pi}}(u \mid x) &\equiv& \sum_y \pi(y \mid x)q^y_{out}(u)
\label{eq:pi_partial}
\ea
(and similarly, for a GQ process over $W$ guided by $Y$).
Conversely, we can always recover $Pr(y \mid x)$ from $Pr(u \mid x)$, simply by marginalizing.
Therefore, we can treat any distribution ${\overline{\pi}}(u \mid x)$ 
defining such a GQ process interchangeably with a distribution $\pi(y \mid x)$
(and similarly, for distributions ${\overline{\pi}}(w \mid y)$ and $\pi(x \mid y)$
occurring in GQ processes over $W$ guided by $Y$).

\begin{enumerate}
\item To construct the GQ process for the first stage, begin by writing:
\ba
\rho^t(w, u) 
 &=&  \sum_{x, y} \G_t(x) \delta(y, 0) q^x_{proc}(w) q^y_{out}(u) \nonumber \\
 &=& q^0_{out}(u) \G_t({\cal{X}}(w)) q^{{\cal{X}}(w)}_{proc}(w) 
\ea
where $\G_t(x)$ is an assumption for the initial distribution over $x$, one that in general may be wrong.
Furthermore, define the associated distribution:
\ba
\rho^t(u \mid x) 
 &=& \frac{ \sum_{w \in {\cal{X}}(x)} \rho^t(w, u) } { \sum_{u', w \in {\cal{X}}(x)} \rho^t(w, u')} \nonumber \\
 &=& q^0_{out}(u)
\ea

By Corollary~\ref{prop:22}, running a GQ process over $Y$ guided by $X$
for conditional distribution ${\overline{\pi}}(u \mid x_t)$ and initial distribution $\rho^t(w, u)$ will
send any initial distribution $\PP_t(x) \rho^t(u \mid x) = \PP_t(x) q^0_{out}(u)$ to a distribution 
$\PP_t(x) {\overline{\pi}}(u \mid x)$. Therefore, in particular, 
it will send any initial $x \rightarrow {\overline{\pi}}(u \mid x)$. Due to the definition of $q^y_{out}$ and~Equation (\ref{eq:pi_partial}), the associated
conditional distribution over $y$ given $x$, $\sum_{u \in {\cal{Y}}(y)} {\overline{\pi}}(u \mid x)$, is
equal to $\pi(y \mid x)$. Accordingly, this GQ process
implements the first stage of the organism process, as desired. In~addition, it preserves the validity of
our assumptions that $Pr(u \mid y) = q^y_{out}(u)$ and similarly for $Pr(w \mid x)$.

Next, by the discussion at the end of Section~\ref{sec:gq_proc_def}, this GQ process will be thermodynamically reversible since
by assumption, $\rho^t(u \mid x) $ is the actual initial distribution over $u$ conditioned on $x$. 

\item To construct the GQ process for the second stage, start by defining an initial distribution based on a
(possibly counterfactual) prior $\G_t(x)$:
\ba
\hat{\rho}(w_t, u_{t}) &\equiv& \sum_{x,y} \G_t(x) q^{x}_{proc}(w_t) {{\pi}}(y \mid x) q^y_{out}(u_{t})
\ea
and the associated conditional distribution:
\ba
\hat{\rho}(w_t \mid y_t) &=& \frac{ \sum_{u_t \in {\cal{Y}}(y_t)} \hat{\rho}(w_t, u_t) }{ \sum_{w', u' \in {\cal{Y}}(y_t)} \hat{\rho}(w', u')}
\ea
Note that:
\ba
\hat{\rho}(w_t \mid y_t)  &=& \G_t(x_t \mid y_t) q^{x_t}_{proc}(w_t)
\ea
where:
\ba
\G_t(x_t \mid y_{t}) &\equiv& \frac{\pi(y_{t} \mid x_t) \G_t(x_t)}{\sum_{x'} \pi(y_{t} \mid x') \G_t(x')}
\label{eq:70} 
\ea
Furthermore, define a conditional distribution: 
\ba
{\overline{\pi}}(w_t \mid y_t) &\equiv& I(w_t \in \X(0)) q^0_{proc}(w_t)
\ea


Consider a GQ process over $W$ guided by $Y$ for conditional distribution ${\overline{\pi}}(w_t \mid y_t)$ and initial distribution
$\hat{\rho}(w_t, u_t)$. By Corollary~\ref{prop:22}, this GQ process implements the second stage, as desired. In addition, it preserves the validity of
our assumptions that $Pr(u \mid y) = q^y_{out}(u)$ and similarly fo~$Pr(w \mid x)$.

Next, by the discussion at the end of Section~\ref{sec:gq_proc_def}, this GQ process will be thermodynamically reversible \emph{if} 
$\hat{\rho}(w_t \mid y_{t+1})$ is the actual distribution over $w_{t}$ conditioned on $y_{t+1}$. 
By~Equation (\ref{eq:70}), this in general requires that $\G_t(x_t)$, the 
assumption for the initial distribution over $x_t$ that is built into
the step (ii) GQ process, is the actual initial distribution over $x_t$. 
As discussed at the end of Section~\ref{sec:q_proc_def}, work will be dissipated if this is not the case.
Physically, this means that if the device implementing this GQ process is thermodynamically optimal for one input
distribution, but used with another, then work will be dissipated
(the amount of work dissipated is given by the change in the Kullback--Leibler divergence 
between $G$ and $\PP$ in that stage (4) GQ process; see~\cite{wolpert_landauer_2016a}).


\item We can also implement the fourth stage by running a (different) GQ process over 
$X$ guided by $Y$. This GQ process is a simple copy operation, \textit{i.e.}, implements a single-valued,
invertible function from $y_{t+1}$ to the initialized state $x$. Therefore, it is thermodynamically reversible. 
Finally, we can implement the fifth stage by running an appropriate GQ process over $Y$ guided by $X$. This
process will also be \mbox{thermodynamically reversible.}
\end{enumerate}
%

\bibliographystyle{amsplain}

\renewcommand\bibname{References}

\end{document}